%% file: main.tex
\documentclass[conference,cmex10]{IEEEtran}
\input{tex/top}

\ifCLASSINFOpdf
\else
\fi
\ifCLASSOPTIONcompsoc
 \usepackage[caption=false,font=normalsize,labelfont=sf,textfont=sf]{subfig}
\else
 \usepackage[caption=false,font=footnotesize]{subfig}
\fi
\begin{document}
% \title{Bare Demo of IEEEtran.cls\\ for IEEE Conferences}
\ifextendedversion
\title{\sys: Protecting ASLR Against Microarchitectural Attacks (Extended Version)}
\else
\title{\sys: Protecting ASLR Against Microarchitectural Attacks}
\fi

\author{\IEEEauthorblockN{Shixin Song}
	\IEEEauthorblockA{Massachusetts Institute of Technology\\
		shixins@mit.edu}
	\and
	\IEEEauthorblockN{Joseph Zhang}
	\IEEEauthorblockA{Massachusetts Institute of Technology\\
		jzha@mit.edu}
	\and
	\IEEEauthorblockN{Mengjia Yan}
	\IEEEauthorblockA{Massachusetts Institute of Technology\\
		mengjiay@mit.edu}}
	

%Telephone: (800) 555--1212, Fax: (888) 555--1212}


\IEEEoverridecommandlockouts
\makeatletter\def\@IEEEpubidpullup{6.5\baselineskip}\makeatother
\IEEEpubid{\parbox{\columnwidth}{
		Network and Distributed System Security (NDSS) Symposium 2025\\
		24-28 February 2025, San Diego, CA, USA\\
		ISBN 979-8-9894372-8-3\\
		https://dx.doi.org/10.14722/ndss.2025.240264\\
		www.ndss-symposium.org
}
\hspace{\columnsep}\makebox[\columnwidth]{}}

% https://dx.doi.org/10.14722/ndss.2025.24xxxx

\maketitle

\input{tex/abstract}
% \begin{abstract}
% The abstract goes here.
% \end{abstract}


% \fi
\IEEEpeerreviewmaketitle

\input{tex/intro}
\input{tex/background}
\input{tex/analysis}
\input{tex/design}
\input{tex/eval}
\input{tex/related}
\input{tex/conclusion}

\input{tex/acknowledgment}



%\bibliographystyle{IEEEtran}

\bibliographystyle{IEEEtranS}
\bibliography{refs}

\appendices
\input{tex/artifact}
\input{tex/appendix}

\ifextendedversion
\input{tex/proof}
\fi

\end{document}

%% file: tex/top.tex
\usepackage[sort,compress]{cite}
\usepackage[T1]{fontenc}
\usepackage{url}

\usepackage{tikz}
\usepackage{fontawesome}

\usepackage[cmex10]{amsmath}

\usepackage{amssymb}
\usepackage{xspace}
\usepackage{physics}
\usepackage{siunitx}

\makeatletter
\ifcase\pdf@shellescape
  \usepackage[frozencache=true,cachedir=minted-cache]{minted}\or
  \usepackage[finalizecache=true,cachedir=minted-cache]{minted}\or
  \usepackage[frozencache=true,cachedir=minted-cache]{minted}\fi
\makeatother
\setminted{linenos,breaklines,breakanywhere,tabsize=2,gobble=2,escapeinside=||,fontsize=\footnotesize,xleftmargin=\parindent,numbersep=6pt}

\usepackage[caption=false, font=footnotesize]{subfig}

\newcommand{\BashFancyFormatLine}{%
  \def\FancyVerbFormatLine##1{\$\ \,##1}%
}
\setminted[bash]{
    linenos=false,
    formatcom=\BashFancyFormatLine,
    xleftmargin=0pt,
    xrightmargin=5pt,
}

\usepackage{enumitem}

\usepackage{multirow}
\usepackage{comment}
\usepackage[normalem]{ulem}
\usepackage{adjustbox}
\usepackage{diagbox}
\usepackage{mathpartir}
\usepackage{extpfeil}
\usepackage{makecell}
\usepackage[english]{babel}
\usepackage{amsthm}
\usepackage{enumitem}

\usepackage{tabularray}
\usepackage{booktabs}
\UseTblrLibrary{booktabs}

\usepackage{caption}
\newtheorem{theorem}{Theorem}
\newtheorem{lemma}{Lemma}
\newtheorem{definition}{Definition}

\usepackage{amsmath}
\usepackage{sansmath}
\SetSymbolFont{largesymbols}{sans}{OMX}{iwona}{m}{n}
\usepackage{filecontents}
\usepackage{inconsolata}
\theoremstyle{remark}

\usepackage{wasysym}

\newif\ifdraft
\draftfalse
\newif\ifextendedversion
\extendedversiontrue

\newcommand{\linux}{Linux 6.6\xspace}

\newcommand{\sys}{\emph{Oreo}\xspace}

\newcommand{\qangle}[1]{\left\langle#1\right\rangle}

\newcommand{\codestart}{\texttt{start}}
\newcommand{\codeend}{\texttt{end}}
\newcommand{\codelen}{\texttt{len}}
\newcommand{\codealign}{\texttt{len}_{\textrm{subregion}}}

\newcommand{\loadoffset}{\texttt{offset}}

\newcommand{\outeroffset}{\loadoffset_{\textrm{oreo}}}
\newcommand{\oreooffset}{\loadoffset_{\textrm{oreo}}}

\newcommand{\oreobits}{microarchitecture oblivious bits\xspace}
\newcommand{\OreoBits}{Microarchitecture Oblivious Bits\xspace}
\newcommand{\oreobit}{microarchitecture oblivious bit\xspace}

\definecolor{myblue}{HTML}{0066CC}
\definecolor{mygreen}{HTML}{009900}

\newcommand{\addrmask}{\mathop{\texttt{Virt2Mask}}}
\newcommand{\addrvirt}{\mathop{\texttt{Mask2Valid}}}

\newcommand{\uarchobs}{\mathop{\texttt{Obs}_{\mu}}}

\newcommand{\pubeq}{=_{\mathrm{pub}}}
\newcommand{\progeq}{=_{func}}
\newcommand{\maskeq}{=_{mask}}

\newcommand{\noneop}{\mathop{\mathbf{none}}}

\newcommand{\fetchop}{\mathop{\mathbf{fetch}}}
\newcommand{\loadop}{\mathop{\mathbf{load}}}
\newcommand{\storeop}{\mathop{\mathbf{store}}}
\newcommand{\checkop}{\mathop{\mathbf{check}}}

\newcommand{\progreq}{\mathop{\texttt{Req}}}
\newcommand{\update}{\mathop{\texttt{Update}}}
\newcommand{\ptw}{\mathop{\texttt{Ptw}}}
\newcommand{\trans}{\mathop{\texttt{Trans}}}
\newcommand{\getoffset}{\mathop{\texttt{Offset}}}
\newcommand{\init}{\mathop{\texttt{Init}}}

\newcommand{\vistart}{\texttt{vstart}_{\mathrm{inst}}}
\newcommand{\viend}{\texttt{vend}_{\mathrm{inst}}}
\newcommand{\pistart}{\texttt{pstart}_{\mathrm{inst}}}
\newcommand{\piend}{\texttt{pend}_{\mathrm{inst}}}

\newcommand{\highlightmath}[1]{\colorbox{blue!20}{\ensuremath{#1}}}

\newcommand{\syslayout}{\mathcal{L}_{\sys}}

\newcommand*\circled[1]{\tikz[baseline=(char.base)]{\node[shape=circle,draw,inner sep=0pt] (char) {#1};}}

\newcommand{\FTdir}{}
\def\FTdir(#1,#2,#3){%
  \FTfile(#1,{\footnotesize{\color{black!40!white}\faFolderOpen}\hspace{0.2em}#3})
  (tmp.west)++(0.8em,-0.4em)node(#2){}
  (tmp.west)++(1.5em,0)
  ++(0,-1.3em) 
}

\newcommand{\FTfile}{}
\def\FTfile(#1,#2){%
  node(tmp){}
  (#1|-tmp)++(0.6em,0)
  node(tmp)[anchor=west,black]{\tt \footnotesize #2}
  (#1)|-(tmp.west)
  ++(0,-1.2em) 
}

\newcommand{\FTroot}{}
\def\FTroot{tmp.west}

\newcommand{\pgheading}[1]{\noindent\textbf{#1.}}
\newcommand{\pgheadingnoperiod}[1]{\noindent\textbf{#1}}

\newcommand{\speccpioverhead}{$0.11\%$\xspace}
\newcommand{\specmaskratiomin}{$99.46\%$\xspace}
\newcommand{\lebenchoverheadavg}{$-0.28\%$\xspace}
\newcommand{\lebenchoverheadmin}{$-6.39\%$\xspace}
\newcommand{\lebenchoverheadmax}{$6.16\%$\xspace}
\newcommand{\lebenchmaskratioavg}{$75.17\%$\xspace}
\newcommand{\lebenchmaskratiomin}{$65.97\%$\xspace}
\newcommand{\lebenchmaskratiomax}{$82.86\%$\xspace}
\newcommand{\repeatnum}{$16$\xspace}

%% file: tex/abstract.tex
\begin{abstract}
Address Space Layout Randomization (ASLR) is one of the most prominently deployed mitigations against memory corruption attacks.
ASLR randomly shuffles program virtual addresses to prevent attackers from knowing the location of program contents in memory.
Microarchitectural side channels have been shown to defeat ASLR through various hardware mechanisms.
We systematically analyze existing microarchitectural attacks and identify multiple leakage paths. Given the vast attack surface exposed by ASLR, it is challenging to effectively prevent leaking the ASLR secret against microarchitectural attacks.

Motivated by this, we present \sys, a software-hardware co-design mitigation that strengthens ASLR against these attacks.
\sys uses a new memory mapping interface to remove secret randomized bits in virtual addresses before translating them to their corresponding physical addresses.
This extra step hides randomized virtual addresses from microarchitecture structures, preventing side channels from leaking ASLR secrets. 
\sys is transparent to user programs and incurs low overhead.
We prototyped and evaluated our design on Linux using the hardware simulator gem5.
\end{abstract}

%% file: tex/intro.tex
\section{Introduction}
\label{sec:intro}
Memory corruption vulnerabilities are some of the oldest security problems that continue to pose a serious security threat to modern systems~\cite{szekeres2013sok, mitre}.
Among all the memory safety mechanisms proposed in the last few decades, Address Space Layout Randomization (ASLR)~\cite{team2003pax, edge2013kernel}, has shown to be effective in raising the barrier of attacks and has become one of the most prominently deployed mitigations in modern systems.
ASLR works by randomly arranging the positions of code or data regions for the
kernel or user-space applications.
If the attackers cannot reliably determine the location of specific code or data, they will have difficulty carrying out control-flow hijacking attacks, such as return-oriented programming~\cite{roemer2012return} and jump-oriented programming~\cite{bletsch2011jump}.
With ASLR, the attacker must perform an extra information disclosure step utilizing other existing vulnerabilities to leak ASLR secret before conducting their exploits.

However, ASLR has been defeated with various microarchitectural side channels.
They pose a real threat, as we witness an increasing number of such attacks being utilized in real-world software exploitations.
For example, in 2017, a macOS kernel 0-day exploit~\cite{Siguza2017IOHIDeous}
used the prefetch attack~\cite{gruss2016prefetch} to bypass ASLR.
In 2022, a Linux kernel exploit (CVE-2022-42703)~\cite{jenkins2022exploiting}
also used side channels to bypass ASLR.
In the same blog post, the authors stated that
``KASLR is comprehensively compromised on x86 against local attackers, and has been for the past several years, and will be for the indefinite future.''

Among these microarchitectural-attack-assisted ALSR bypasses, a wide range of channels can be utilized, including TLBs, caches, and BTBs, using speculative execution, or even power-induced timing information~\cite{gras2017aslr,zhao2022binoculars,gras2018translation,koschel2020tagbleed,gruss2016prefetch,liu2023entrybleed,evtyushkin2016jump,evtyushkin2018branchscope,lee2017inferring,goktas2020speculative,jang2016breaking,hund2013practical,canella2020kaslr,schwarz2019store,canella2019fallout,weber2021osiris,lipp2022amd}.
Even worse, given the large attack surface, it seems that almost every newly discovered side-channel attack variant will likely become a new ASLR-bypassing attack vector.
Given this phenomenon, two questions present themselves: 
(1) why has ASLR become such a fragile target for microarchitectural attacks, and
(2) how can we secure ASLR to broadly block existing and potential future attack vectors?

\pgheading{Challenges}
With a detailed investigation of existing ASLR bypasses using microarchitectural side channels, we find that the microarchitecture features that can be leveraged to use as a leakage channel are diverse and continue growing.
Therefore, when protecting ASLR against side-channel attacks, addressing each individual channel or feature is not an appealing approach.
For example, FLARE~\cite{canella2020kaslr} narrowly focuses on closing a single channel, i.e., the address translation latency, while ignoring the abundant other side channels in modern processors, as well as other existing attack vectors.
Instead, in this paper, we focus on blocking the leakage at the source by restricting the usage of the ASLR secret in both software and hardware.

Consider how the ASLR secret is used.
ASLR shifts the location of a memory region by a \textit{secret offset}.
As such, the secret offset determines the virtual memory layout, i.e., which memory region is mapped and which is unmapped.
The secret offset is also embedded in code and data pointers that are extensively used while executing a victim program.
We systematically analyze and categorize microarchitectural ASLR bypasses into three leakage paths.

In the first leakage path, the attacker probes the virtual memory layout.
This class of attacks relies on the fact that probing a mapped address versus an unmapped address resolves different microarchitectural side effects and thus distinct latency~\cite{lipp2022amd,goktas2020speculative,hund2013practical,jang2016breaking,wikner2023phantom,canella2020kaslr,weber2021osiris,schwarz2019store,canella2019fallout}.
Additionally, these probe operations can be stealthy and do not cause system crashes, because they are either conducted under speculation or assisted with cache manipulation instructions such as \texttt{prefetch}.

Second, as the ASLR secret offset is embedded in code and data pointers, the attacker can leak the secret by monitoring the victim using its secret-dependent pointers to fetch instructions or perform loads and stores~\cite{lipp2020take,gras2017aslr,zhao2022binoculars,gras2018translation,koschel2020tagbleed,gruss2016prefetch,liu2023entrybleed,evtyushkin2016jump}.
These operations can result in distinguishable side effects on BTB, TLB, page table walker, caches, and DRAM.

Finally, the attacker can leak the victim's pointers by using the Spectre attack gadget and its variants.
For example, the attacker may load the secret pointer into a register and then use it as the address of a load or a store instruction.
There exists a substantial amount of work to block this leakage path, such as STT~\cite{yu2019speculative}, NDA\cite{weisse2019nda}, InvisiSpec~\cite{yan2018invisispec}, and others~\cite{ainsworth2021ghostminion,choudhary2021speculative,barber2019specshield,fustos2019spectreguard,koruyeh2020speccfi,schwarz2020context,yu2018data,daniel2023prospect,loughlin2021dolma,mosier2023serberus,yu2020speculative,ainsworth2020muontrap,khasawneh2019safespec,kiriansky2018dawg,li2019conditional,saileshwar2019cleanupspec,sakalis2019ghost,sakalis2019efficient}.
% \pwd{I would merge the paragraph below with the one above}
Unfortunately, existing mitigations targeting Spectre and its variants can only block the third leakage path, leaving the other two leakage paths unblocked.
Moreover, many mitigations work by selectively delaying secret-dependent speculative execution.
Such schemes are applicable to the backend of a processor (at the load/store queue) with moderate performance loss, but they are unappealing to the frontend of the processor, where delaying fetching instructions means leaving the rest of the processor seriously underutilized.

\begin{figure}
    \centering
    \includegraphics[width=\linewidth]{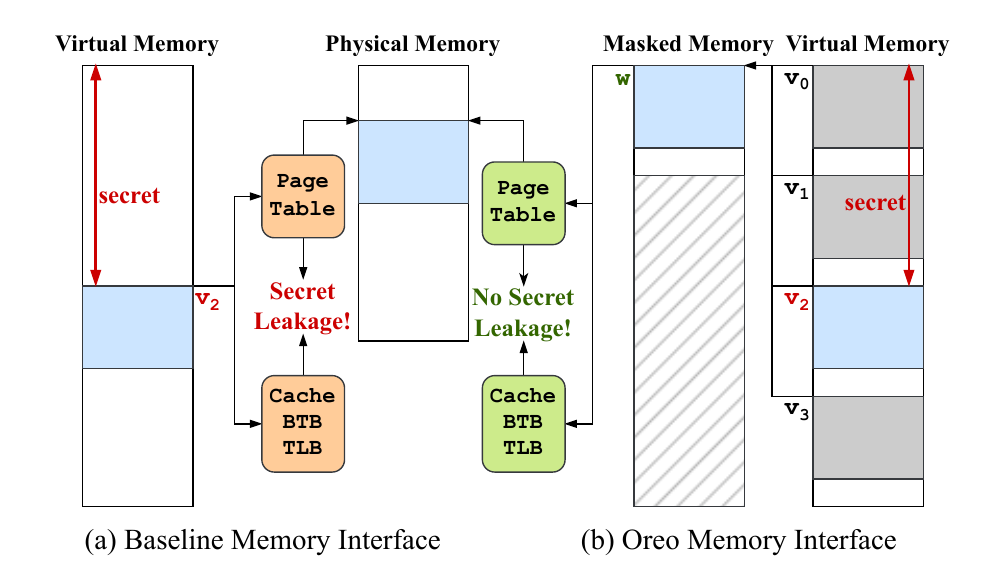}
    \caption{Overview of \sys's new memory interface}
    \label{fig:sys-mem-layout-overview}
\end{figure}

\pgheading{\sys}
The challenges discussed above motivate us to design \sys, a software-hardware co-design scheme that secures ASLR against wide range of microarchitectural side-channel attacks.
Specifically, given a randomized virtual address, \sys aims to protect selected randomized bits from being leaked via microarchitectural attacks following the first two leakage paths.
We refer to these protected bits as \emph{\oreobits}.

The core of \sys is a new memory interface, as shown in Figure~\ref{fig:sys-mem-layout-overview}.
The left side of the figure shows the existing memory interface utilizing ASLR, mapping randomized virtual addresses to physical addresses.
A randomized virtual address, which embeds the secret ASLR offset, is used as input to lookup entries in various hardware structures (caches, BTB, TLB, etc.) and critical software structures (the page table).

In contrast, \sys introduces a new layer of memory, called \emph{the masked memory}, sitting between the virtual and physical memory, shown on the right side of the figure.
A masked address is constructed from a randomized virtual address with the \oreobits being redacted.
In this way, \sys can map multiple virtual addresses to the same masked address and then maps masked addresses to physical addresses using a modified page table.
For example, in Figure~\ref{fig:sys-mem-layout-overview}, we show a valid region in the virtual address space starting with address $v_2$. 
\sys maps $v_2$ and another three invalid addresses $v_0,v_1,v_3$ to the same masked address $w$.
All the hardware structures that used to use virtual addresses as inputs, now switch to using masked addresses.

\sys additionally changes the memory security check procedure.
Given mapped and unmapped virtual addresses, which only differ in the \oreobits, accessing them on \sys will have the exact same microarchitectural side effects during speculation and only result in different architectural behaviors upon instruction commit time.

We further identify research challenges in selecting which bits to be protected by \sys.
In fact, not all the ASLR randomized bits can become \oreobits.
Moreover, it poses a security dilemma to obtain the entropy towards locating gadgets for both control-flow hijacking attacks and speculative execution attacks.
We provide bits selection strategies to achieve entropy towards both attacks, and meanwhile to be adoptable by existing systems.

We prototype our design with support for protecting both kernel and user-space ASLR.
We integrate the kernel changes on \linux, and 
implement our microarchitecture changes on the gem5~\cite{Binkert:2011:gem5,Lowe-Power:2020:gem5-20} simulator. 
We show that our design introduces negligible performance overhead running the SPEC2017 IntRate benchmark~\cite{bucek2018spec} and the LEBench benchmark~\cite{zzrcxb-LEBench-Sim}.
\ifextendedversion
We provide a formal proof in Appendix~\ref{sec:proof} to show that \sys achieves a non-interference property to prevent attackers from distinguishing virtual memory layouts with different secret offsets.
\else
We provide a formal proof in a technical report~\cite{song-oreo-proof} to show that \sys achieves a non-interference property to prevent attackers from distinguishing virtual memory layouts with different secret offsets.
\fi

In summary, we make the following contributions:
\begin{itemize}
[leftmargin=*]
    \item We systematically analyze existing microarchitectural attacks that leak the ASLR secret offset and classify them into three leakage paths.
    \item We propose a software-hardware co-design mitigation to prevent leaking selected bits of the ASLR secret, innovating a new memory interface. 
    % Our scheme can support both coarse-grained and page-granularity fine-grained ASLR schemes.
    \item We prototype the software and hardware changes to support \sys on Linux and the gem5 simulator.
    \item 
    \ifextendedversion
    We provide security evaluation and a formal proof to show that \sys prevents leakage of the ASLR secret, and our performance evaluation shows \sys introduces negligible overhead.
    \else
    We provide security evaluation to show that \sys prevents leakage of the ASLR secret, and our performance evaluation shows \sys introduces negligible overhead.
    \fi
\end{itemize}

%% file: tex/background.tex
\section{Background}
\label{sec:background}

\subsection{Address Space Layout Randomization}
\label{sec:background:aslr}
Address Space Layout Randomization (ASLR)~\cite{team2003pax,edge2013kernel} 
is a widely deployed mitigation against memory corruption attacks.
The idea is to randomly relocate code or data regions so that the attacker has difficulty determining the addresses for specific instruction gadgets to construct reliable code-reuse attacks.
Figure~\ref{fig:kaslr-mem} describes how coarse-grained ASLR works.

Given the full virtual address space, ASLR selects a memory region to perform its random re-location.
In the whole virtual address space, multiple non-overlapping randomization regions exist for relocating different memory contents.
For example, Linux kernel uses a region to hold kernel text and a different region for kernel modules.
Figure~\ref{fig:kaslr-mem}(a) shows two randomization regions denoted as $[\codestart, \codeend)$ and $[\codestart', \codeend')$.

Within each randomization region, only a subset of address slots are used to hold code/data, thus considered valid addresses (the blue regions), while the rest are considered invalid.
The length of the valid region is usually significantly smaller than the total size of the randomization region, which is necessary for achieving a reasonable amount of entropy.
Address relocation shifts the valid region by an $\loadoffset$, which is the distance between the starting address of the valid region and the randomization region. This offset can be a public value when no ASLR is used (Figure~\ref{fig:kaslr-mem}(b), and must be chosen secretly when ASLR is in place (Figure~\ref{fig:kaslr-mem}(c)).

\begin{figure}[t]
    \centering
    \includegraphics[width=\linewidth]{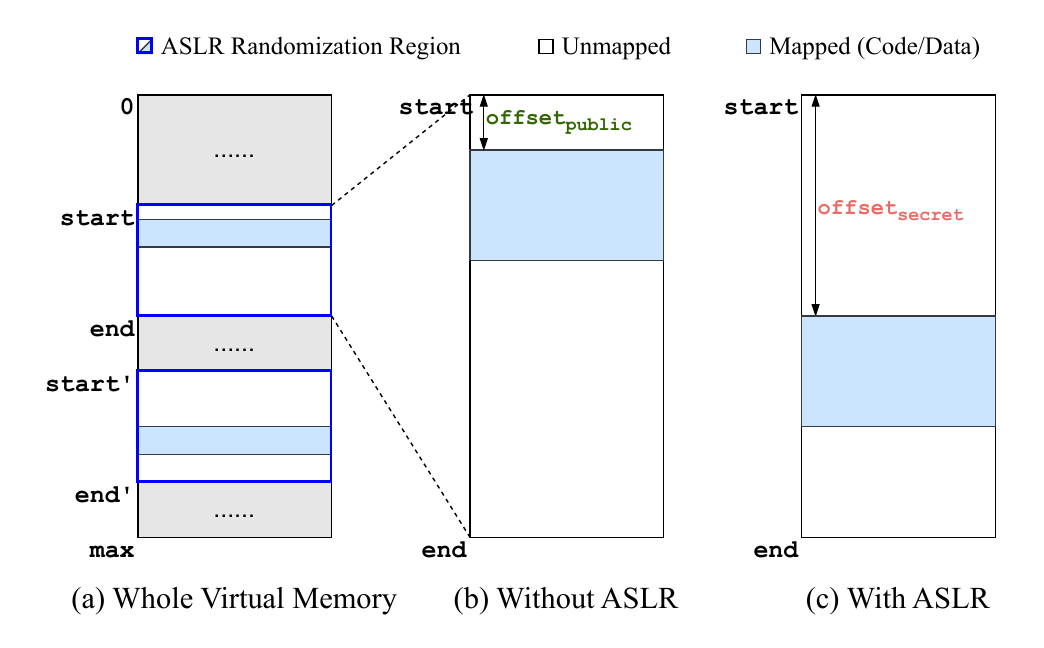}
    \caption{Coarse-grained ASLR.
    (a) shows the whole virtual address space with multiple randomization regions.
    (b) and (c) show memory content is loaded using a public offset when ASLR is disabled and a private offset when ASLR is enabled.}
    \label{fig:kaslr-mem}
\end{figure}

\subsection{Virtual Memory Systems}
\label{sec:background:memory_system}

Modern systems support virtual memory for the purpose of process isolation, programmability, and hardware abstraction.
With the virtual memory interface, the software does not directly operate on physical addresses backed by DRAM. 
Instead, the operating system abstracts DRAM by providing software with a large, contiguous, and unified virtual address space, and introduces a layer of indirection to translate every virtual address to its mapped physical address.

Linux and many existing operating systems use page-based address translation.
Taking the virtual page number (VPN) from a virtual address, \textit{a page table walk} translates the VPN to its corresponding physical page number (PPN).
Modern systems use hierarchical page tables to store the page table in a space-efficient manner.
The CPU looks up virtual addresses by traversing a tree structure from root to leaf, requiring multiple memory accesses.
Microarchitecture structures, such as translation lookaside buffers (TLBs) and page table caches, can buffer recently accessed translations to accelerate this procedure.

The virtual memory system enforces security checks upon every memory access leveraging protection bits embedded in page table entries (PTEs).
The check involves checking whether a virtual address is mapped or not. 
Accessing an unmapped address leads to looking up an invalid PTE entry and thus results in a page fault. 
The virtual memory security check additionally examines whether the access to the page has the correct permissions from an appropriate privilege level. 

\subsection{Microarchitectural Side Channel Attacks}
\label{sec:background:side_channel}

A microarchitectural side channel involves information leakage from a victim's security domain to an attacker's security domain.
The attacker exploits visible side effects of the execution of instructions whose behaviors are secret-dependent. 
We call such instructions \textit{transmitters}.
The transmitter leaves side effects by modifying the states or occupancy of various microarchitecture structures, such as caches~\cite{kiriansky2018speculative,Kocher2018spectre,koruyeh2018spectre,Lipp2018meltdown,maisuradze2018ret2spec,schwarz2019netspectre}, TLBs~\cite{gras2018translation,koschel2020tagbleed}, BTB~\cite{evtyushkin2016jump,chowdhuryy2021leaking}, and Network-on-Chips~\cite{dai2022don}.
Furthermore, recent microarchitectural attacks exploit power-induced timing leakage~\cite{wang2022hertzbleed}.

Speculative execution attacks, also referred to as transient execution attacks, are a class of information leakage attacks where attackers exploit the side effects of \textit{transient} instructions.
A transient instruction is an instruction that is speculatively executed on an out-of-order core but is later squashed due to misspeculation.
High-profile speculative execution attacks include Meltdown~\cite{Lipp2018meltdown}, Spectre~\cite{Kocher2018spectre}, and its variants~\cite{kiriansky2018speculative,koruyeh2018spectre,maisuradze2018ret2spec,schwarz2019netspectre}.

Given that modern processors have been aggressively optimized, an increasing number of microarchitectural side channels and speculative execution features have been revealed in the last few years.
We envision this trend will continue.
Furthermore, as the arms race continues,
due to performance overhead and hardware costs, it is unlikely that the industry will commit to delivering processors with comprehensive mitigations in the near future.
Hence, non-complete solutions that only block certain leakage channels are insufficient. 
This paper presents a mitigation that blocks microarchitectural-attack-assisted ASLR bypasses.
Our scheme works even if the processor exhibits vulnerable microarchitectural leakage channels and speculation features.

%% file: tex/analysis.tex
\section{Understanding ASLR Bypasses Using Microarchitectural Attacks}
\label{sec:attack-analysis}
ASLR can be bypassed by leaking the secret offset using either software attacks or microarchitectural attacks.
In this paper, we focus on microarchitectural-attack-assisted ASLR bypasses because most of these attacks do not need to exploit software bugs and can universally work with a wide range of commercial processors.
In fact, these attacks have been put into use in real-world software exploits~\cite{ravichandran2022lord,Siguza2017IOHIDeous,jenkins2022exploiting}. 

We systematically analyze the past microarchitectural-attack-assisted ASLR bypasses and find there exist multiple possible leakage paths.
As the ASLR secret is spread throughout the system, it unavoidably exposes a large attack surface.
Specifically, the secret offset determines the virtual memory layout, meaning it determines which regions are
mapped and which are not.
The secret offset is also embedded in code pointers, which are extensively used by the processor to fetch instructions when executing a victim program.

Figure~\ref{fig:aslr-attack-graph} summarizes these attack surfaces.
From top to bottom, we show where the secret offset is located (the top row), what attacker operations (the second row) can be used to trigger a secret-dependent transmission operation (the third row), and what microarchitectural side channels the transmission operation can modulate (the bottom row).
Overall, the figure shows three leakage paths.
Moreover, the microarchitectural channels (structures) that can be used to leak the secret are diverse, indicating blocking each single side channel is not a promising strategy to defend against ASLR bypass attacks.
A detailed enumeration of existing attacks is provided in Appendix~\ref{sec:full-attack-summary}, including the leakage paths they take and the utilized side channels.

\begin{figure}
    \centering
    \includegraphics[width=\linewidth]{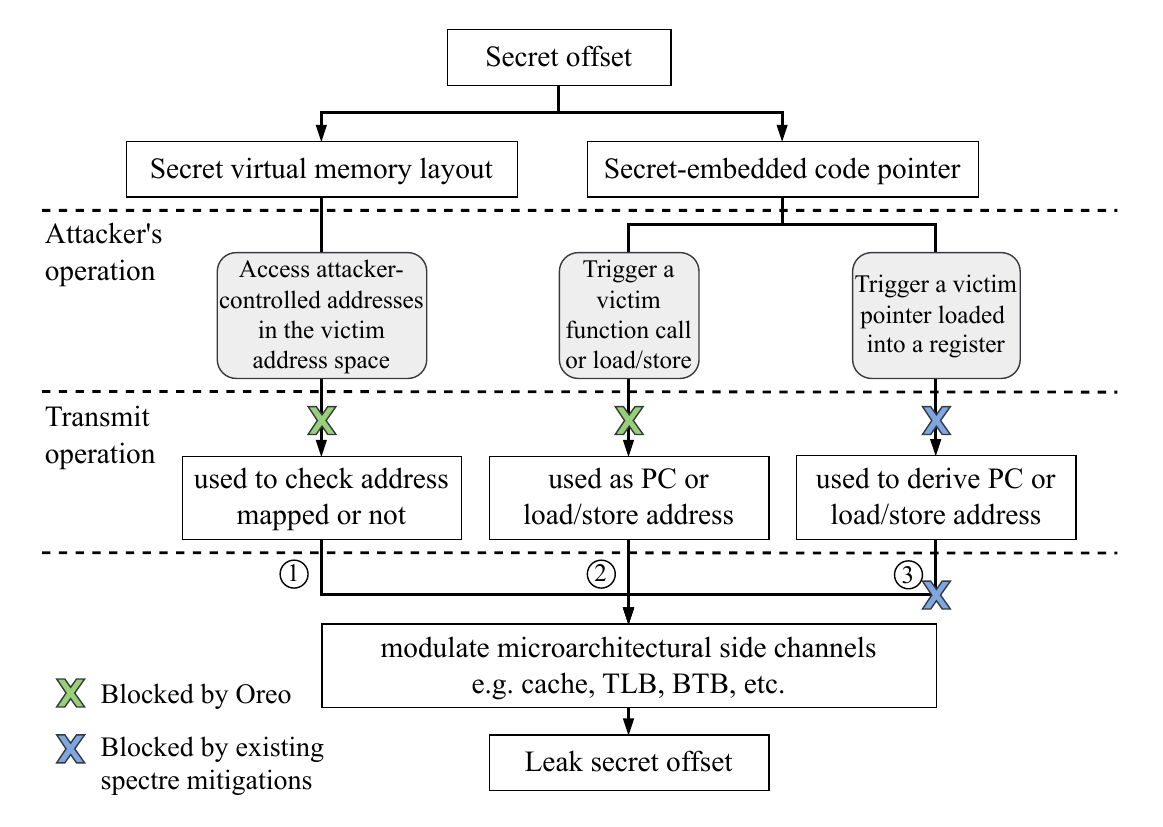}
    \caption{Three leakage paths that leak ASLR secret offset}
    \label{fig:aslr-attack-graph}
\end{figure}

\pgheading{Virtual Memory Layout Probing}
The first leakage path (the left path in Figure~\ref{fig:aslr-attack-graph}) shows an attacker probes the virtual memory layout to figure out which region is mapped and which is unmapped~\cite{lipp2022amd,goktas2020speculative,hund2013practical,jang2016breaking,wikner2023phantom,canella2020kaslr,weber2021osiris,schwarz2019store,canella2019fallout}.
The attacker's operation involves triggering the processor to access an attacker-controlled address in the victim address space.
The processor needs to consult several microarchitecture structures and perform a page table walk if needed to determine whether the address is mapped or not, resulting in distinguishable microarchitectural side effects.
Usually, accessing unmapped addresses will result in system crashes or exceptions.
However, microarchitectural attacks can be stealthy since a clever attacker can suppress the crashes using speculation and other tricks, 
such as cache manipulation instructions (e.g., prefetch) and Intel TSX~\cite{intel2022intel}.

Multiple existing attacks fall into this category.
For example, as shown in the following code snippet, DrK~\cite{jang2016breaking} attacks kernel ASLR by probing each page in the kernel's randomization region from the user space and checking whether the page is mapped or not.
\begin{minted}{c++}
// guess_addr is a kernel address, controlled by the attacker
for (guess_addr = start; guess_addr < end;
     guess_addr += page_size) {
  // Step 1: Access guess_addr, suppress crashes with Intel 
  // TSX. TLB caches transformation if guess_addr is valid.
  tsx_probe(guess_addr);
  // Step 2: Access guess_addr again, measure page fault 
  // latency. Latency is low if the translation is in TLB.
  page_fault_latency = tsx_probe(guess_addr);
  if (page_fault_latency < threshold) { break; }
}
\end{minted}
For each page, the attacker probes (e.g., issue a load) to a virtual address (denoted as $\texttt{guess\_addr}$) in that page two times. In both instances, the attacker uses Intel TSX to suppress page faults caused by failed permission checks.
Specifically, instead of informing the OS and causing a real crash, Intel TSX invokes a user-specified exception handler when a page fault occurs~\cite{intel2022intel}.
At the first probe, a page table walk (PTW) is triggered.
The PTW inserts a PTE entry into the TLB if the tested page is mapped and leaves the TLB unchanged otherwise.
The attacker then issues the second probe to test whether the address translation for $\texttt{guess\_addr}$ is cached in the TLB.
This test is performed by measuring how long it takes for the second probe to trigger a page fault that is then caught by the attacker's exception handler.

As shown above, attacks in this leakage path rely on distinguishing between mapped and unmapped addresses based on their microarchitectural side effects.
In addition to using TLB behaviors, prior work~\cite{canella2020kaslr, weber2021osiris, schwarz2019store, canella2019fallout} also discovered that pipeline behaviors differ for mapped and unmapped addresses.
Taking the Data Bounce attack~\cite{schwarz2019store, canella2019fallout} as an example,
they exploit the timing difference introduced by the store-to-load forwarding scheme, which is only triggered when the store address is mapped.

\pgheading{Leaking Pointers as Addresses}
The middle path in Figure~\ref{fig:aslr-attack-graph} shows how the secret offset embedded in a victim pointer can be leaked when it is used as an \textit{address} during the victim's execution.
Specifically, the pointer is used as the program counter or a load/store address.
Unlike the first leakage path, this attack vector does not require speculation or any crash suppression schemes. 
For example, to use a victim address as the program counter,
the attacker only needs to trigger the victim to do a function call \textit{non-speculatively}.
When accessing a secret-dependent address, various microarchitecture structures will be modulated, including BTBs (Jump Over ASLR~\cite{evtyushkin2016jump}), TLBs (TLBBleed~\cite{gras2018translation},  TagBleed~\cite{koschel2020tagbleed}, the Prefetch attack~\cite{gruss2016prefetch}, and EntryBleed~\cite{liu2023entrybleed}), and page table walkers (AnC~\cite{gras2017aslr} and Binoculars~\cite{zhao2022binoculars}).

We provide an example below to illustrate the AnC attack~\cite{gras2017aslr}, which uses cache Prime+Probe to leak kernel ASLR from user space.
After resetting the cache states, the attacker simply makes a system call, which triggers the victim (i.e., the kernel in this example) to fetch instructions using secret-dependent virtual addresses and the processor will trigger the page table walker (PTW) to translate these virtual addresses.
Note that, as the PTW uses secret bits of the virtual addresses to index into page tables and modulate the caches, the attacker can leak the secret offset by monitoring cache states. 
\begin{minted}{c}
// Step 1: reset cache states
prime_the_whole_cache();
// Step 2: trigger the victim to fetch mapped virtual 
// addresses and the PTW process to translate these addresses.
syscall();
// Step 3: Use cache Prime+Probe to figure out which cache set
// has been modulated by the PTW. "hit_set" reveals  
// selected bits in the virtual addresses used by PTW
hit_set = probe_every_cache_set();  

\end{minted}

\pgheading{Leaking Pointers as Data}
The final leakage path (shown as the right path in Figure~\ref{fig:aslr-attack-graph}) describes how the secret embedded victim pointers can be leaked as data.
Different from the previous leakage path, the attacker needs to leverage a memory corruption vulnerability or transient out-of-bound memory access to load a victim pointer into a register and then use the secret bits in the pointer to compute an address.
This derived address is then leaked via side channels.

This leakage path usually requires a classic Spectre attack gadget or its variants~\cite{Kocher2018spectre,kiriansky2018speculative,koruyeh2018spectre,maisuradze2018ret2spec,schwarz2019netspectre}.
Such gadgets pose a serious threat as they can be exploited to leak arbitrary data in the victim's address space, not just pointers. 
As such, extensive work has been proposed to mitigate this threat.
For example, STT~\cite{yu2019speculative} and NDA~\cite{weisse2019nda} delay speculative execution of transmission instruction (e.g., load/store) if their operand holds speculative data.
Many other prior works~\cite{yan2018invisispec, ainsworth2021ghostminion,choudhary2021speculative,barber2019specshield,fustos2019spectreguard,koruyeh2020speccfi,schwarz2020context,yu2018data,daniel2023prospect,loughlin2021dolma,mosier2023serberus,yu2020speculative,ainsworth2020muontrap,khasawneh2019safespec,kiriansky2018dawg,li2019conditional,saileshwar2019cleanupspec,sakalis2019ghost,sakalis2019efficient} intend to hide side effects of speculative execution on the cache hierarchy and other structures.

\section{Threat Model}
\label{sec:threat_model}
We follow the threat model of microarchitectural side channel attacks, where the attacker and the victim reside in different security domains.
This setup includes the case when the attacker and the victim are two user-space processes, or the victim is a privileged software, such as an operating system kernel and hypervisor, while the attacker is a user-space application.
This also applies when the victim is an enclave program while the attacker is privileged software.
The attacker and the victim execute on the same machine, sharing various microarchitecture structures.
We broadly consider timing-based side channels due to resource contention and speculation.
Our threat model does not specifically consider power-induced timing side channels~\cite{lipp2022amd}, but we elaborate on how our scheme can help future mitigations in Section~\ref{sec:related}.

We set out to block the first two ASLR leakage paths in Figure~\ref{fig:aslr-attack-graph}, indicated by the two green crosses. 
We propose \sys, a software-hardware co-design scheme to prevent leakage through direct virtual memory layout probing and ensure secret-embedded pointers, when used as program counters or load/store addresses, remain indistinguishable to attackers. 
\sys is highly practical and can be adopted in real systems.

As a side note, substantial prior work~\cite{yu2019speculative,weisse2019nda,ainsworth2021ghostminion,choudhary2021speculative,barber2019specshield,fustos2019spectreguard,koruyeh2020speccfi,schwarz2020context,yu2018data,daniel2023prospect,loughlin2021dolma,mosier2023serberus,yu2020speculative,ainsworth2020muontrap,khasawneh2019safespec,kiriansky2018dawg,li2019conditional,saileshwar2019cleanupspec,sakalis2019ghost,sakalis2019efficient} has attempted to address the data leakage path (the third path).
These schemes can be used in complementary with \sys to serve as a comprehensive defense solution for achieving security goals beyond mitigating ASLR bypasses. 

%% file: tex/design.tex
\section{Design of \sys}
\subsection{Overview}
The goal of \sys is to protect selected randomized bits in virtual addresses from being leaked via microarchitectural attacks in the first two leakage paths in Figure~\ref{fig:aslr-attack-graph}.
We refer to these bits as \textit{microarchitecture oblivious bits} or \textit{protected bits} for short.
The protected bits in a valid randomized virtual address concatenated with trailing zeros form the secret offset. 
Without ambiguity, we refer to the secret offset protected by \sys as $\oreooffset$.
For example, given a valid randomized address \texttt{0xFF\textcolor{mygreen}{AB}12340}, if configuring the protected bits as bits $20$ to $27$, then the secret $\oreooffset$ is $\texttt{0x\textcolor{mygreen}{AB}00000}$.

\pgheadingnoperiod{How to Protect \OreoBits?}
We introduce a new memory interface, with an extra layer of masked address space that sits between the virtual address space and the physical address space.
A masked address is mapped from a virtual address with the \oreobits redacted.
As shown in Figure~\ref{fig:sys-mem-layout-detail}, multiple virtual addresses with different \oreobits are mapped to the same masked address.
\sys uses masked addresses to build and traverse page tables and access various microarchitecture structures, such as BTBs and TLBs.
This scheme ensures accessing a virtual address results in microarchitectural side effects independent from its protected bits. 

In addition, \sys changes the memory address security check flow.
Using masked addresses makes unmapped and mapped virtual addresses that only differ in \oreobits have exactly the same microarchitectural side effects, but we eventually need to distinguish between them to raise exceptions upon illegal memory accesses.
\sys performs this check at the commit time of instructions, forcing the distinguishability between these addresses to happen only at the architectural level rather than the microarchitectural level.

\pgheadingnoperiod{How to Choose Which Bits to Protect?}
There exist several constraints in choosing the bits to be protected by \sys.
For example, \sys's protection faces a security dilemma.
On the positive side, \sys prevents leaking $\oreooffset$ using microarchitectural attacks, strengthens the security of ASLR, and raises the barrier against control-flow hijacking attacks.
However, on the negative side, an attacker who aims to perform speculative execution attacks can execute a Spectre-like gadget under speculation without knowing the secret $\oreooffset$.
This is because accessing addresses only differing in the \sys protected bits speculatively results in the same microarchitectural side effects.

We present two strategies to determine which bits to be protected by \sys.
A naive strategy chooses to protect part of the default baseline ASLR randomized bits, and an enhanced strategy introduces additional entropy and protects bits that do not overlap with the randomized bits in the baseline.

\subsection{Protecting \OreoBits}
\label{sec:design-protect}
\begin{figure}
    \centering
    \includegraphics[width=\linewidth]{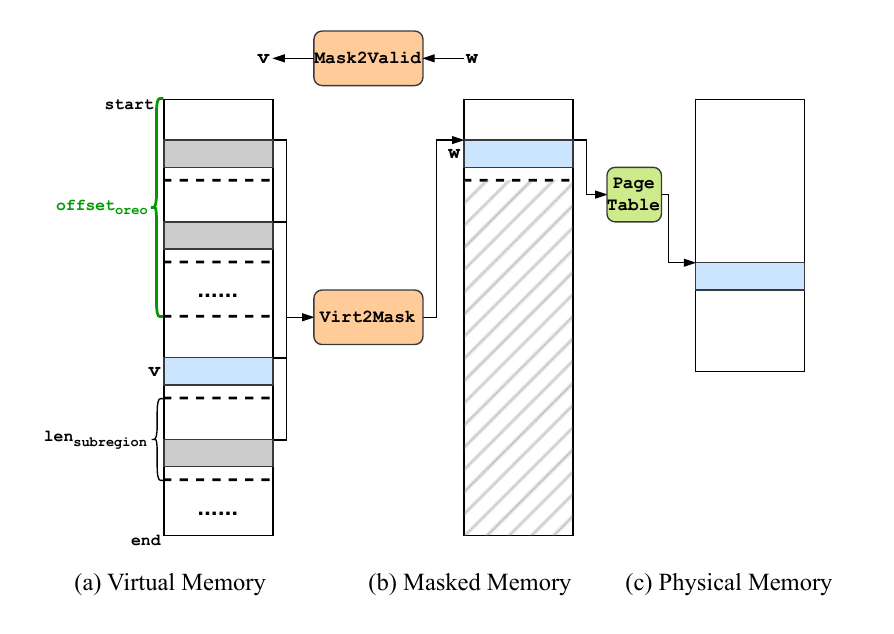}
    \caption{Mapping between a virtual address $v$ and a masked address $w$ using \sys's memory interface.}
    \label{fig:sys-mem-layout-detail}
\end{figure}

\pgheading{Masked Address Space}
\label{sec:masked-address-space}
\sys introduces the masked address space and maps multiple virtual addresses to the same masked address.
Figure~\ref{fig:sys-mem-layout-detail} shows the notation we use to describe the masked address space layout.
Given a memory range to be used for ASLR denoted as $[\codestart, \codeend)$, we divide this range into multiple subregions with equal size denoted as $\codealign$, marked by the dotted lines in Figure~\ref{fig:sys-mem-layout-detail}(a).
Addresses in the $i$-th region have their \oreobits equal to $i$, and the goal is to prevent these bits from being leaked.

\sys maps each subregion to the very first subregion in the masked address space.
Given a valid address $v$, the corresponding masked address can be calculated as:
\begin{equation*}
        \addrmask(v) = ((v - \codestart) \bmod \codealign) + \codestart.
        % \codestart + (v - \codestart) \bmod \codealign.
\end{equation*}
We note two things from the formula above.
First, the calculation of the virtual-to-masked memory mapping, as well as the masked-to-physical mapping, are independent of \oreobits.
Second, it is simple enough to be implemented efficiently in hardware with low cost.

After the address conversion is done, \sys needs to set up page tables to map masked addresses instead of virtual addresses. 
We discuss required kernel changes in Section~\ref{sec:implementation-memory-interface}.
In addition, any microarchitecture structures that use virtual addresses as indices will be modified to use masked addresses, as detailed in Section~\ref{sec:hw-changes}.

\pgheading{Security Check Flow}
\sys forces the virtual addresses, as long as they map to the same masked address, to be only distinguished at the architectural level after the instruction commits.
In other words, accessing unmapped addresses eventually results in exceptions upon instruction commits.
We clarify how \sys's security check flow interacts with the existing virtual memory security check.

\sys's security check flow involves two steps.
The first step is performed during speculation and is exactly the same as the existing virtual memory security check, with the only difference being performing the check upon masked addresses rather than virtual addresses.
Recall from Section~\ref{sec:background:memory_system}; the existing virtual memory security check examines whether the accessed page is mapped as well as whether the access has the correct read/write/execute permissions from an appropriate privilege level.
This check is performed during TLB accesses or page table walks.
\sys applies this check on masked addresses during speculation so that this check is independent from the \oreobits.

The second step of the security check happens at the commit time.
The task is to check whether the virtual address has the correct \oreobits.
This can be done by reconstructing the valid virtual address and checking whether the virtual address used by the committing instruction is equal to the valid virtual address.
The following formula reconstructs the valid virtual address from a masked address $w$.
\begin{equation*}
    \addrvirt(w) = \outeroffset + w.
\end{equation*}
This formula indicates that we need to store the secret $\oreooffset$ somewhere to be readily used upon instruction commit time.
In our implementation, we store this information in the unused bits in page table entries~\cite{intel2022intel}, which we detail in Section~\ref{sec:implementation-memory-interface}.

\pgheading{Blocking Leakage Paths}
We show below how \sys blocks the first two leakage paths using concrete examples. Formal reasoning is presented in Appendix~\ref{sec:proof}.
Consider the following example where a program is loaded at a random address \texttt{0xFF\textcolor{mygreen}{AB}12340}.
If \sys intends to protect $\oreooffset=\texttt{0x\textcolor{mygreen}{AB}00000}$, the corresponding masked address can be computed as \texttt{0xFF0012340}.

In the first leakage path, the attacker probes the victim's virtual address space by issuing speculative memory access operations targeting different victim addresses and expecting to distinguish between mapped and unmapped addresses based on their microarchitectural side effects.
For example, when performing the Double Page Fault attack~\cite{hund2013practical} and Code Region Probing attack~\cite{goktas2020speculative} on the baseline insecure system, if the attacker accesses any address with incorrect \oreobits (e.g., $\texttt{0xFF\textcolor{red}{AA}12340}$), the TLB will not be filled for this invalid address translation, and the caches will not be filled with such an invalid address.
This differs from the microarchitectural behaviors when accessing the valid address.

Using masked addresses allows \sys to have the exact same microarchitectural effects when accessing addresses that differ in \oreobits for two reasons.
First, on \sys, virtual addresses that fall within the randomization will be converted to masked addresses before address translation and memory accesses.
Both the invalid address $\texttt{0xFF\textcolor{red}{AA}12340}$ and the valid one \texttt{0xFF\textcolor{mygreen}{AB}12340} will be converted to the same masked address \texttt{0xFF0012340}.
In both cases, the TLB will be filled with a valid address translation, and a valid cache entry located using the translated physical address will be inserted into the cache hierarchy.
Second, \sys delays the check until the commit time to determine whether an address is mapped or not.
Recall that, as the attacker aims to suppress potential exceptions using pipeline squashes, none of the attacker's probing operations will reach the commit stage.
As a consequence, they follow the same security check under speculation and their microarchitectural effects will not be distinguishable on \sys.

Reasoning about how \sys blocks the second leakage path is straightforward.
Recall that in the second leakage path, the attacker triggers a victim functional call to make the victim branch to a valid virtual address (e.g., \texttt{0xFF\textcolor{mygreen}{AB}12340}).
In the insecure baseline, the secret bits \texttt{0x\textcolor{mygreen}{AB}} are used as part of the address to locate entries in TLBs, BTBs, and caches. 
While, \sys redacts the secret bits from the virtual address to obtain the secret-free masked address \texttt{0xFF\textcolor{mygreen}{00}12340}.
Furthermore, all the microarchitecture structures that used to take virtual addresses as input now switch to using masked addresses, leading to no side effects related to the secret $\oreooffset$, so that \sys successfully blocks the second leakage path.

\subsection{Choosing Bits to Protect}
\label{sec:design-dilemma}
\begin{figure*}[t]
    \centering
    \includegraphics[width=\linewidth]{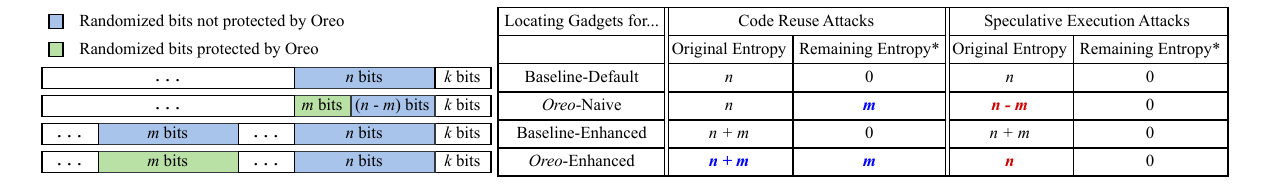}
    \caption{Protected bits selection strategies and their corresponding entropy. ``Remaining entropy*'' refers to ``remaining entropy after ASLR bypasses using the first two leakage paths''}
    \label{fig:design-addr-para}
\end{figure*}

\pgheading{Constraints}
Before presenting our strategies in choosing the bits to be protected by \sys, let's first understand what constraints we have.

First, the least significant \oreobit is constrained by the subregion size used by \sys.
Recall from Figure~\ref{fig:sys-mem-layout-detail}, \sys divides the memory range to be used by ASLR into multiple equally-sized subregions and maps these subregions to a single subregion in the masked address space.
It is required that only one of the subregions in the virtual address space is valid.
As such, we need to ensure the subregion size $\codealign$ to be large enough to hold the code or data that is being relocated.
For example, the Linux kernel text is $2^{25}=\SI{32}{MB}$, and thus the least significant bit we can pick to be protected by \sys is bit $25$.
The baseline ASLR does not have this constraint since the valid region can be shifted at the page granularity.
The least significant bit to be randomized is bit $12$ if using $\SI{4}{KB}$ pages or bit $21$ if using $\SI{2}{MB}$ pages.

Second, we face a security dilemma.
\sys makes mapped and unmapped addresses to be indistinguishable at the microarchitectural level if they only differ in \oreobits.
The problem is that if attackers aim to conduct speculative execution attacks, they can utilize Spectre-like gadgets without knowing the correct \oreobits in the gadgets' randomized virtual addresses.
The attack can succeed because accessing virtual addresses with incorrect \oreobits has the same microarchitectural effects during speculation as accessing the one with the correct bits.

If we increase the bits to be protected by \sys, we increase the system's resilience against ASLR bypass attacks and thus make control-flow hijacking attacks, such as ROP attacks, more difficult.
However, we may face the risk of decreasing the time it takes an attacker to locate and exploit Spectre gadgets.
Therefore, we need to choose which bits to protect carefully to optimize for the entropy we can achieve against both attacks.

\pgheading{Bits Selection Strategies}
Figure~\ref{fig:design-addr-para} lists the ASLR randomized bits and \oreobits in four different setups,
including the default baseline ASLR used in Linux, a naive \sys bits selection strategy, an enhanced baseline with additional randomized bits, and an enhanced \sys selection strategy.
In the baseline ASLR setup used by Linux (top row in the figure), the $n$ bits in the virtual address, colored in \textcolor{myblue}{blue}, are randomized, and the lower $k$ bits are not randomized. 

The naive strategy is to choose part of the ASLR secret bits as the \oreobits.
For example, we can choose the higher $m$ bits (colored in \textcolor{mygreen}{green}) of the ASLR secret bits, so the least significant bit protected by \sys is the $(k+n-m)$th bit.
Due to the subregion size constraint discussed before, we have to ensure $m$ is small enough so that we have $2^{k+n-m}$ to be at least as large as the valid region size.
The concern with this strategy is when the value of $m$ is too small, \sys has limited entropy towards mitigating ASLR bypasses.
Besides, it is not effective in addressing the speculative execution security dilemma discussed above either.

We address the above constraints by proposing an enhanced bits selection strategy, which introduces extra entropy in addition to the entropy of the default ASLR.
For a fair comparison, we present an enhanced baseline (the third row) and the enhanced \sys strategy (the bottom row) in Figure~\ref{fig:design-addr-para}.
The enhanced baseline additionally randomizes the higher $m$ bits in the virtual address to the left of the $n$ bits already randomized by the default ASLR.
\sys's enhanced bits selection strategy randomizes the same $m+n$ bits as the enhanced baseline, but protects the higher $m$ bits (marked as \textcolor{mygreen}{green}).
In this case, we no longer need to be concerned with the subregion size constraint since the least significant bit protected by \sys is not smaller than $(k+n)$.
$2^{k+n}$ is the default ASLR randomization region and is deemed to be larger than the valid region size.

\pgheading{Entropy Analysis}
We summarize the entropy comparison of these schemes in Figure~\ref{fig:design-addr-para}.
We show the entropy of each scheme against locating gadgets to be used in code reuse attacks and speculative execution attacks.
We then compare the original entropy of each scheme with their remaining entropy after the attacker performs any of the microarchitectural-based ASLR bypasses in the first two leakage paths in Section~\ref{sec:attack-analysis}.
Note that blocking speculative execution attacks is out of the scope of our threat model; rather, we want to ensure \sys does not make this type of attack easier.

We begin by clarifying the security implications of the original entropy.
It is desired to retain the original entropy of \sys to match the corresponding baseline ASLR.
In this way, we force the attacker to pay extra effort to perform ASLR bypass attacks.
However, further increasing the original entropy does not help strengthen the security of ASLR against microarchitectural attacks. 
Several attacks in leakage path \circled{2}, such as the AnC attack~\cite{gras2017aslr} and Binoculars~\cite{zhao2022binoculars}, directly leak all ASLR randomized bits.
Importantly, their leakage time is independent of the number of bits to be leaked.
As shown in Figure~\ref{fig:design-addr-para}, though the enhanced baseline has higher original entropy than the default baseline, they both have no resilience against microarchitectural-attack-assisted ASLR bypasses and have $0$-bit remaining entropy against locating gadgets for both code reuse attacks and speculative execution attacks.

With the naive \sys bit selection strategy, we obtain $m$-bit entropy against gadget detection for code reuse attacks.
The original entropy of gadget detection for speculative execution attacks is reduced from $n$ bits to $n\!-\!m$ bits since the attacker does not need to know the \oreobits.
Using the enhanced bits selection scheme, we have $m$-bit remaining entropy against gadget detection for code reuse attacks and $n$-bit original entropy against gadget detection for speculative execution attacks.

To summarize, introducing extra randomized bits that are not protected by \sys does not gain security against microarchitectural attacks.
Even though the enhanced baseline provides $m$ more bits of original entropy compared to the default baseline and the \sys enhanced strategy, it increases little security guarantee against speculative execution attacks.
Overall, the \sys enhanced bits selection strategy achieves the best security property,
effectively increasing the barrier against control-flow hijacking attacks and retaining the barrier against speculative execution attacks.

\pgheading{Feasibility and Linux Prototyping}
It is feasible to adopt the enhanced \sys bit selection strategy in existing systems.
We have implemented this strategy in our Linux prototype for kernel text, kernel modules, and user-space programs.

For kernel text, Linux's default configuration relocates the code within a 1GB region using 2MB as the relocation alignment.
As such, the default ASLR randomizes bits $21$ to $29$, providing a $9$-bit entropy.
Similarly, ASLR relocates kernel modules within a 1GB region that is consecutive to the kernel text randomization region and uses 4KB as the alignment size, so bits $12$ to $29$ are randomized. To reserve large enough memory for kernel modules, the default ASLR allows $1024$ possible offsets for relocation, providing $10$-bit entropy.

Following the enhanced strategy in Figure~\ref{fig:design-addr-para}, we need to randomize additional bits higher than the ASLR secret bits.
According to the Linux kernel memory management documentation~\cite{kernelmemorymanagement}, there exists a consecutive $\SI{444}{GB}$ unused region that can serve our randomization goal. 
We use this region for both kernel text and modules, so each of them can use $\SI{222}{GB}$.
We configure \sys to protect bits 31 to 38 (8 bits in total) for both the kernel text and kernel modules.

We can also apply the enhanced selection strategy to kernel data regions and user-space memory.
For example, the default user-space ASLR can use the whole user-space virtual address space as the randomization region with a granularity of $\SI{4}{KB}$ and it provides a high entropy of $28$ bits.
To maximize applicability, we do not want to reduce the available virtual memory size for the user space.
Therefore, we choose non-canonical bits as \oreobits.
Considering a memory system using 4-level page tables, a canonical address is derived by taking a 48-bit virtual address and sign-extending it to form a 64-bit address.
Bits 48 and above are not used in the baseline, so using them as \oreobits 
will not affect available memory size.

\subsection{Further Increasing Entropy of \sys-Protected Bits}
\label{sec:design-fine-grained}
So far, we have shown \sys works with coarse-grained ASLR, where the addresses from the same valid region (the blue region in Figure~\ref{fig:sys-mem-layout-detail}(a)) share the same secret offset.
However, one limitation of coarse-grained ASLR is that leaking one pointer breaks the whole defense. 
We now examine how to further increase the entropy of \sys-protected bits.

\pgheading{Working with Existing FGASLR}
The existing fine-grained ASLR randomizes the memory layout by (1) shuffling the order of functions inside the program and (2) relocating the program by a random offset (same as the coarse-grained ASLR). 
\sys can only protect the random offset in (2) but cannot prevent leaking the order of functions after shuffling.

\pgheading{Supporting Page-Granularity ASLR}
An appealing feature of \sys is its capability to conveniently support page-granularity ASLR to improve its safety level.
Figure~\ref{fig:page-level-fine-grained} shows a configuration where we have four pages in the valid region.
Figure~\ref{fig:page-level-fine-grained}(a) describes when using coarse-grained ASLR, the four pages use the same $\outeroffset$ to be relocated to the third subregion in the virtual address space.
Alternatively, Figure~\ref{fig:page-level-fine-grained}(c) shows the case when the four pages are relocated to different subregions in the virtual address space while still mapped to non-overlapped pages in the masked address space.

We now analyze the $\addrmask$ and $\addrvirt$ functions to highlight that minimal changes are needed to make \sys support page-granularity ASLR.
First, the coarse-grained and the fine-grained schemes use the same $\addrmask$ function.
Second, the $\addrvirt$ function only differs slightly as the page-granularity scheme needs to configure different $\outeroffset$ for different pages, which means storing different $\outeroffset$ values in the corresponding PTE entries.

Configuring \sys is handy if page-granularity ASLR is in place.
However, we acknowledge there exist engineering challenges to implementing a proper code relocator for such a randomization scheme in software.

\begin{figure}[t]
    \centering    
    \includegraphics[width=\linewidth]{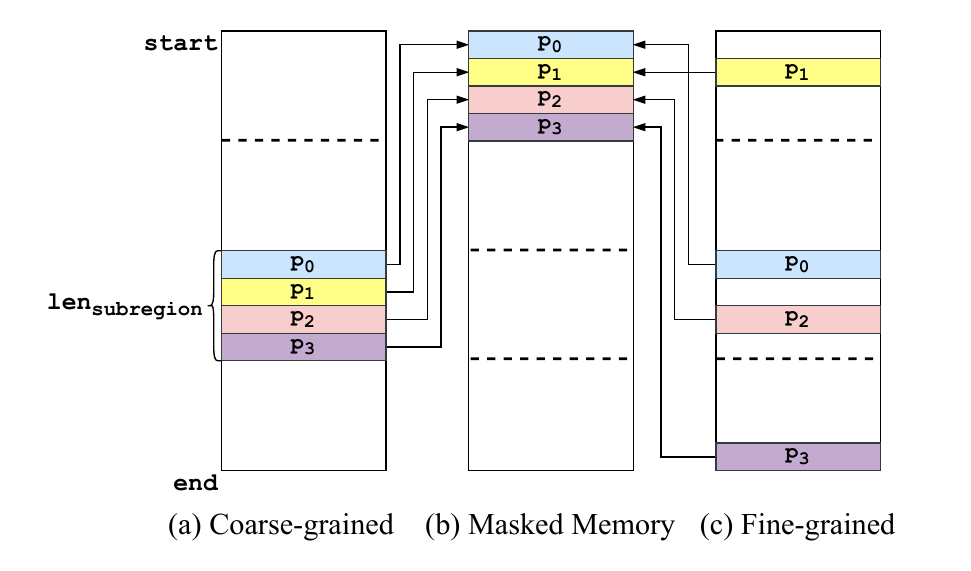}
    \caption{Compare \sys's randomized memory layout for coarse-grained ASLR (left) and page-granularity ASLR (right)}
    % \todo{Change \texttt{align}}
    \label{fig:page-level-fine-grained}
\end{figure}

\subsection{Limitations}
\label{sec:limitations}
While \sys offers substantial improvements in ASLR security, it also has several limitations that must be considered.
In this section, we summarize these limitations to help understand the trade-offs involved in applying \sys.

\pgheading{Security Dilemma}
In Section~\ref{sec:design-dilemma}, we have detailed the security dilemma problem and how we address it by carefully selecting bits to protect.
Basically, on the one hand, \sys successfully protects ASLR randomized bits against the first two leakage paths and makes control-flow hijack attacks more difficult. 
On the other hand, the protected bits cannot prevent attackers from utilizing Spectre gadgets since they can speculatively execute the gadgets without knowing these bits.
When applying \sys, we need to take this dilemma into consideration and carefully choose bits to protect so that the overall system achieves a satisfactory security guarantee against both attacks.

\pgheading{Using Non-Canonical Bits or Canonical Bits}
When applying \sys to user-space programs, we choose \oreobits from non-canonical bits.
We see no problem in applying such a scheme in systems using 4-level page tables where we implement our prototype, but it may introduce limitations due to reduced non-canonical bits in other system configurations.

For example, given a system using 5-level page tables and 57-bit virtual addresses, \sys can use no more than $7$ bits as \oreobits.
The non-canonical bits might also be used for other purposes such as ARM PAC~\cite{inc2017armpac},
which further limits the entropy. % of \oreobits.
To achieve higher entropy, we can use canonical bits as \oreobits for user-space programs, following a similar configuration used by \sys for kernel text and modules.

However, using canonical bits faces another constraint due to the program's code or data sizes.
As analyzed in Section V-C, the least significant \oreobit is constrained by the subregion size ($\codealign$) used by \sys, which must be large enough to hold the relocated code or data.
Different programs have varying code and data sizes.
Enforcing a uniform configuration for all user-space programs to accommodate the largest programs would unnecessarily limit the entropy for smaller programs.
Instead, \sys can adopt a more flexible approach.
Specifically, $\codealign$ can be adjusted based on the size of each program, allowing \sys to protect more bits and achieve higher entropy for smaller programs.

\pgheading{Using PTE Bits}
As \sys stores the protected bits in PTEs, the maximum entropy of the protected bits is limited by the number of unused bits in PTEs. 
For example, PTE bits can be used by other system software and security mechanisms such as MPK~\cite{kernelmemoryprotectionkeys}.
In our prototype, we have already considered these factors and used PTE bits that do not overlap with MPK bits and software-available bits.
Specifically, our prototype (Section~\ref{sec:implementation-memory-interface}) uses leaf PTEs to store \oreobits, which allows at most 9 bits for the kernel space and 5 bits for the user space.
If future software uses more leaf PTE bits, \sys can use the unused bits in other levels of PTEs.

\section{Implementation Details}
\label{sec:implementation-details}
\subsection{Software Changes}
\label{sec:implementation-memory-interface}
\sys's software changes require modifying the page table to use masked addresses.
We prototype the software changes in Linux kernel.
Our current prototype supports three types of memory regions: the kernel text, kernel modules, and user-space programs. 
It is also feasible to adopt \sys for other memory regions, including kernel data regions.

For all the randomization regions, we set up the page table to map masked addresses to physical addresses and record the secret $\outeroffset$ in PTE entries.
Our implementation is compatible with the existing Linux kernel implementation.
The implementation for the three regions slightly differs on 1) which bits are selected to be \oreobits, and 2) when a page mapping is set up.

\pgheading{Kernel Text}
Linux kernel uses 21 to 29 as the ASLR randomized bits.
Our prototype uses the enhanced bits selection strategy (Section~\ref{sec:design-dilemma}) to additionally randomize and protect bits 31 to 38, i.e., the \oreobits.
We store these bits at the PTE entry for each kernel text page.

Several Linux-specific implementation details are worth mentioning. 
First, Linux sets up the page table for kernel text at boot time.
Therefore, our prototype directly integrates our changes at the boot code.
Second, since our prototype uses $8$ \oreobits and the x64 architecture reserves $9$ unused bits in leaf PTEs~\cite{accardi2020finer}, our changes to the PTE entries are compatible with existing implementations without affecting the current page tables' functionality. 

\pgheading{Kernel Modules}
Similar to kernel text, we choose bits $31$ to $38$ as the \oreobits for kernel modules.

We note one implementation detail specific to the current implementation of Linux.
In our prototype, we use the same $\oreooffset$ for kernel modules and kernel text.
This differs from the baseline ASLR, where the kernel modules can have different entropy from kernel text.
The reason is that the current Linux implementation requires the kernel text and modules located in a $\SI{2}{GB}$ consecutive region, so bits $31$ to $38$ (the \oreobits) in their virtual addresses need to be the same.
If we want to support different entropy for kernel text and modules, we will need to relax this constraint by leveraging prior work,  such as Adelie~\cite{nikolaev2022adelie}, which allows relocating kernel modules in the whole 64-bit address space.

We make the following changes to the page setup procedure for kernel modules.
Linux allocates memory for a module when the module is loaded into the kernel.
The virtual memory allocation triggers page table entry setup.
Additionally, Linux builds a red-black tree to manage memory for these modules.
In the default implementation, each module's randomized virtual base address is used as the key to construct and search the tree, which introduces a side-channel vulnerability. 
We change the kernel to allocate and manage memory using masked addresses. 
In addition to properly setting the page table entries, we also use modules' masked base addresses to build the red-black tree.

\pgheading{User-Space ASLR}
As discussed in Section~\ref{sec:design-dilemma}, we use the non-canonical bits as \oreobits for user-space applications.
Theoretically, we have $16$ non-canonical bits.
In our prototype, we use bits $48$ to $52$, providing $5$-bit entropy, which matches the number of unused bits in user-space leaf PTEs (rather than $9$ unused bits in kernel leaf PTEs).
We note that this is an engineering decision made for convenience.
It is feasible to increase the entropy by storing the extra \oreobits in the PTEs of the other levels of page tables.
As each user process has its own page tables, we can configure different processes to use different \oreobits and store these bits in the per-process page tables.

\subsection{Microarchitecture Changes}
\label{sec:hw-changes}
One of the central ideas in \sys is to limit the usage of secret-dependent randomized virtual addresses in microarchitecture structures.
We modify the processor pipeline to extensively use masked addresses except for the virtual address security check at the commit stage. 
We  divide the pipeline into three components: the frontend for fetch and decode, the middle component for execute and memory operations, and the backend for committing instructions.
In Figure~\ref{fig:design-uarch}, we use different colors to indicate the usage of different types of addresses in each microarchitecture structure: red for secret bits and secret-dependent virtual addresses, green for secret-free masked addresses, and blue for secret-free physical addresses.

\begin{figure*}[t]
    \centering
    \includegraphics[width=\linewidth]{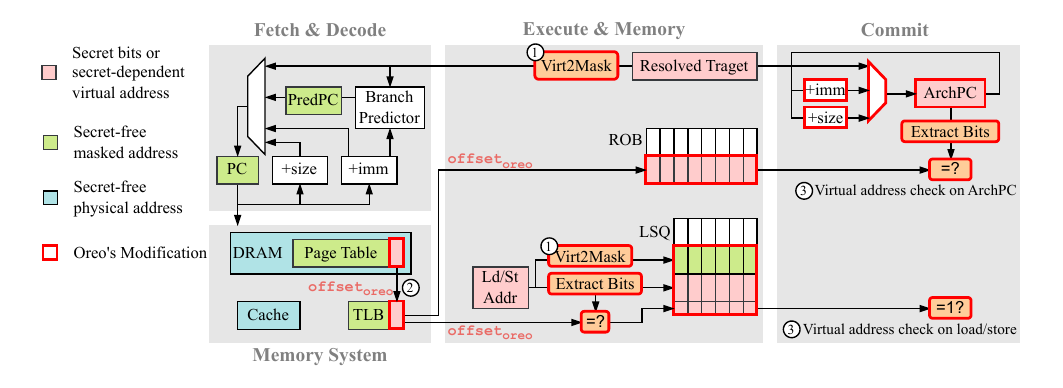}
    \caption{Microarchitecture changes required by \sys.}
    \label{fig:design-uarch}
\end{figure*}

\pgheading{The Fetch \& Decode Stage}
The frontend fetch stage maintains a PC (program counter) register and the branch predictor.
\sys uses the masked address in both the PC register and the branch predictor, requiring no changes to the fetch stage.
To understand why, consider the four sources that the hardware uses to update the PC register:
(1) PC + instruction size for non-branch instruction;
(2) PC + an immediate value from the decode stage for relative direct branches such as \texttt{jmp short};
(3) target address predicted by the branch predictor;
and (4) target address from the register file or memory, usually for indirect jumps, such as \texttt{call} and \texttt{ret}.

Among the four sources above, only the last type obtains the branch target from other pipeline stages, namely the execute \& memory stage.
We place a $\addrmask$ module between the execute \& memory and the fetch stages to ensure the external PC update source uses masked addresses.
As such, all the internal updates of the fetch stage, including the PC-derived targets and predicted targets, will consistently use masked addresses without extra intervention.

\pgheading{The Execute \& Memory Stage}
The middle component of a speculative processor uses a ROB to track all the in-flight instructions, and a load/store queue (LSQ) to track all the in-flight load and store operations.

We extend the ROB (reorder buffer) to store the correct \oreobits (i.e., $\oreooffset$) for the PC of each instruction to facilitate \sys's security check at the commit stage.
Specifically, when fetching instructions using a masked address, the address translation procedure looks up TLBs or performs a page table walk to obtain its physical address.
As we store $\oreooffset$ in PTE entries and the TLB, the translation procedure can obtain $\oreooffset$ and send it back to the core.

In this pipeline stage,
we also prevent information leakage caused by load and store instructions.
When inserting each load and store instruction into the load/store queue (LSQ), we convert the virtual address used by the load/store instruction into a secret-free masked address and also extract the \oreobits.
Memory dependency checks and other microarchitectural optimizations, such as load-to-store forwarding and address prediction, use masked addresses as inputs.
When a load/store is issued to the memory system, similar to instruction fetch, the address translation procedure uses the masked address.
For each load/store instruction, the memory system returns data (if it is a load) and the correct $\oreooffset$ for the given masked address.
This correct $\oreooffset$ is then compared against the extracted \oreobits in the LSQ entry to determine whether the original load/store virtual address is valid. 
The comparison result is stored in the corresponding LSQ entry.
Note that, this operation is secure because this pre-computed check result does not affect any other microarchitecture states, will only be used in the commit stage, and thus does not introduce new timing side channels.

\pgheading{The Commit Stage}
In the baseline hardware, security checks on memory accesses are performed during the address translation within the MMU.
When an instruction enters the commit stage, the hardware checks whether it should trigger an exception (e.g., failing to pass a security check) and retires the instruction if no exception occurs.
In \sys, MMU only performs page permission checks on masked addresses during speculation, i.e., whether the masked address has the correct read/write/execute permission on a given page.
Notably, the MMU leaves the check of whether the \oreobits of the virtual address is correct to the commit stage.

We perform two virtual address checks in the commit stage depending on the instruction types.
For any type of instruction, we check whether the \oreobits of its PC is correct.
For load/store instructions, we additionally perform this check on the load/store addresses.
As we have pre-computed the check result on load/store addresses in the execute \& memory stage and stored the check result in the LSQ, we can directly use this check result in the commit stage.

To check the \oreobits in PCs, we add another PC register to the commit stage.
This PC register keeps track of the virtual address for the instruction at the head of ROB.
To distinguish from the PC register in the fetch stage, we call it \texttt{ArchPC}, since the PC register at the fetch stage is \textit{speculative}, while the one at the commit stage keeps track of the PC to be committed, reflecting architecture-level information.
Since we inserted the correct $\oreooffset$ into the ROB in previous stages, we can conveniently check the validity of \texttt{ArchPC} by extracting \oreobits of the \texttt{ArchPC} and comparing them with the correct bits.
If the \texttt{ArchPC} is valid, we can retire the instruction. Otherwise, an exception is triggered for a virtual address check failure.
Note that we always let the other exceptions take priority over this virtual address check results to avoid leaking ASLR secrets to attacks using microarchitectural replay attacks~\cite{skarlatos2019microscope}.

Furthermore, the commit stage replicates some of the next-PC computation logic similar to the fetch stage.
As shown in Figure~\ref{fig:design-uarch}, this replicated part involves integer addition operators, data forwarding from the execution stage, and a multiplexer.
After successfully retiring an instruction, we use the next-PC logic to derive the architecture PC value of the next instruction and update the \texttt{ArchPC} register.
The design above is a naive baseline to showcase that \sys's hardware modification is not intrusive.
There exist plenty of optimization opportunities to omit these replicated computation structures or move it to an earlier stage of the pipeline.

\subsection{Timing and Hardware Cost Analysis}

We provide a comprehensive analysis to estimate the timing and area overhead.
We acknowledge that a synthesizable implementation of the hardware components of \sys can offer a more accurate measurement.

\pgheading{Pipeline Timing Impacts}
The hardware changes introduced by \sys include:
(1) applying $\addrmask$ to virtual addresses to obtain masked addresses;
(2) obtaining $\oreooffset$ from page tables/TLBs;
(3) checking whether a virtual address has valid \oreobits at the commit stage.
We labeled these changes in Figure~\ref{fig:design-uarch}.

First, the $\addrmask$ operation is performed before branch squashing and issuing load/store requests, labeled as \circled{1} in Figure~\ref{fig:design-uarch}.
Given there can exist multiple randomization regions,
when converting a virtual address to a masked address,
we first determine which randomization region the input virtual address falls into.
This operation requires parallel comparisons between the virtual address and the boundaries of each randomization region.
Once the randomization region is determined, we obtain the information of which bits are \oreobits.
We can then clear these bits using a quick bit-wise \texttt{AND} operation to obtain the masked address.
Overall, it is a lightweight operation consisting of parallel comparisons and bit-wise \texttt{AND} operations. 
Hence, we count no extra cycles introduced by $\addrmask$.

Second, \sys obtains secret ASLR offsets from page tables to fill in TLB entries and further passes them to the ROB and the load/store unit,
labeled as \circled{2} in Figure~\ref{fig:design-uarch}.
These operations are performed in parallel with address translation and do not introduce extra latency.

Third, \sys extracts \oreobits in the \texttt{ArchPC} register and check whether they match the correct bits or not, labeled as \circled{3} in Figure~\ref{fig:design-uarch}.
We count no extra latency or back-pressure at commit time because the information used for the virtual address check on the instruction PC is available at the execution stage, and the logic to do the check is a simple combinational circuit.
Furthermore, the virtual address check results for load/store addresses are pre-computed before commit time as discussed in Section~\ref{sec:hw-changes}.

\pgheading{Area Cost}
We analyze the area cost of \sys's hardware by visiting each block in Figure~\ref{fig:design-uarch}.
To begin with, no extra hardware is needed in the fetch stage as this stage purely operates on masked addresses.

In the memory system, \sys includes additional bits to record the correct $\oreooffset$ for each TLB entry.
Specifically, in our prototyped system, we add $8$ extra bits per TLB entry.
For reference, a Mega BOOM processor~\cite{zhaosonicboom} has $584$ TLB entries (including iTLB, dTLB, and L2TLB), leading to an overhead of $584$ bytes.

In the execute and memory stage of the pipeline, \sys incorporates the $\addrmask$ and bits extraction modules, as well as extra bits in the ROB and LSQ. 
First, both the $\addrmask$ module and the bits extraction module need to
store metadata of randomization regions, including the boundary and \oreobits information. 
For each randomization region, we use $128$ bits to store its boundary (i.e., $\codestart$ and $\codeend$), and a $64$-bit vector to indicate which bits are \oreobits.
Our prototype implementation uses two randomization regions to protect the kernel and user ASLR, which yields a total $384$-bit overhead.
Second, \sys introduces extra fields in the ROB and LSQ to facilitate virtual address checks on instruction PCs and load/store addresses, including
$8$ bits to store the correct $\oreooffset$ for PC in each ROB entry, 
another $8$ bits to store the extracted offset bits of the load/store address in each LSQ entry,
and $1$ bit per LSQ entry to store the pre-computed virtual address check result.
For reference, the Mega BOOM processor has a $128$-entry ROB and a $64$-entry LSQ, resulting in $200$ bytes storage overhead.

In the commit stage, we add the \texttt{ArchPC} register, which is 64-bit.
The bits extraction module requires extra storage for recording randomization region metadata, which is shared with the modules in the execute and memory stage. 

Overall, using the Mega BOOM processor as an example, \sys incurs small storage overhead, 
including $256$ bytes in-core overhead and $584$ bytes overhead in the memory system.

%% file: tex/eval.tex
\section{Evaluation}

\subsection{Experiment Setup}
We implement our kernel changes in \linux. 
We use a kernel patch by Hou et. al.~\cite{hou2023x86} to relocate kernel text and modules to a $\SI{444}{GB}$ unused region in the kernel address space.
We implement our microarchitecture changes in the gem5 simulator (v24.0)~\cite{Binkert:2011:gem5,Lowe-Power:2020:gem5-20} using the full system mode.
The microarchitecture configuration is similar to the configurations used in prior microarchitectural mitigation papers~\cite{yu2019speculative,weisse2019nda,loughlin2021dolma}.
We model a 1-core CPU for running the SPEC2017 benchmark~\cite{bucek2018spec} and security evaluation, and a 2-core CPU for running the LEBench benchmark~\cite{zzrcxb-LEBench-Sim}.
We configure each core as an 8-issue out-of-order (O3) superscalar processor with 32 load queue entries, 32 store queue entries, and 192 ROB entries.
The branch predictor uses the tournament prediction policy with 4096 BTB entries and 16 RAS entries.
We model 64KB 8-way L1 I-cache and D-cache and a 2MB 16-way L2 cache.
gem5 has a customized procedure for booting Linux which deviates from how Linux boots on a real processor and does not support kernel ASLR.
In our implementation, we modify gem5's booting procedure to support kernel ASLR.
Overall, our prototype involves 785 lines of code (LoC) changes to the Linux kernel, and 1897 LoC modification to the gem5 simulator.

\subsection{Performance Results}
\begin{figure}
    \centering
    \includegraphics[width=\linewidth]{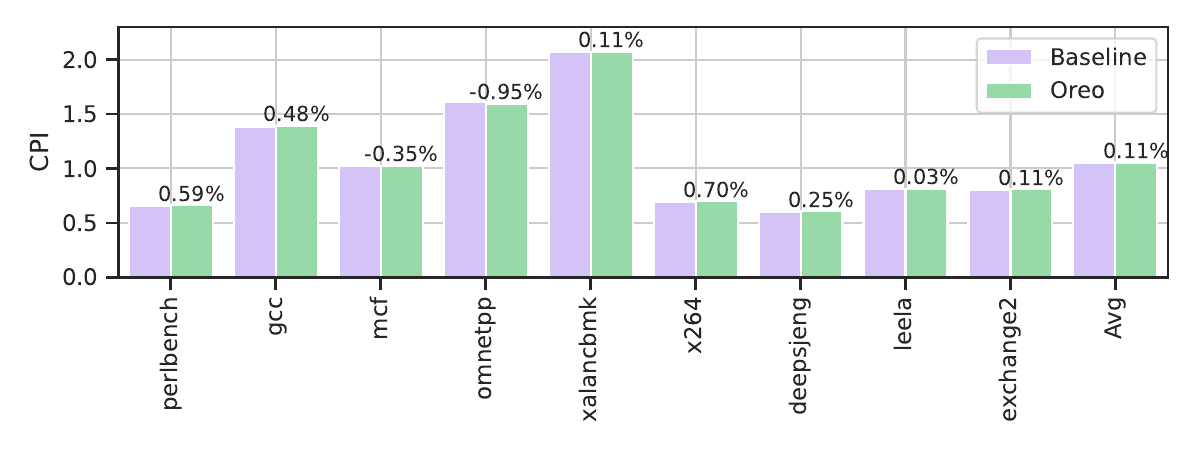}
    \caption{Performance evaluation results on SPEC2017}
    \label{fig:spec}
\end{figure}

\pgheading{SPEC2017}
We evaluate the performance overhead of \sys on the SPEC2017 IntRate benchmark~\cite{bucek2018spec}.
We configure the applications to use reference input size, warm up the system and microarchitecture structures by executing $10$ billion user-space instructions, and then measure the performance of the next $1$ billion instructions.
We skipped \path{557.xz_r} because the simulation crashes due to a bug in gem5 for not supporting certain instructions used by this application.
Figure~\ref{fig:spec} shows the reported CPI (cycles per user-space instruction) for each application. 
We label the CPI overhead ratios incurred by \sys compared to the baseline on top of the green bars.
The overall CPI (counting both user-space and kernel instructions) is mostly identical to the CPI for userspace only, given that the userspace time dominates when running the SPEC benchmark.

Across all applications, \sys introduces negligible performance overhead compared to the baseline, incurring \speccpioverhead CPI overhead on average.
This indicates that \sys's changes on the software and hardware have little impact on the overall performance of user-space applications.

We report the ratio of memory accesses where \sys applies its protection to provide insight into how much of the program execution triggers protection.
This information allows us to validate the relevance of the SPEC benchmark in evaluating the performance impacts of \sys.
In our prototype, we apply $\addrmask$ to any addresses that fall within the user address space, or the selected randomization region for kernel code or modules.
For the SPEC benchmark, at least \specmaskratiomin of the memory accesses trigger \sys's protection, with this being the minimum ratio observed among all applications.

\pgheading{LEBench}
We also evaluate the performance overhead of \sys using the LEBench benchmark~\cite{zzrcxb-LEBench-Sim}, a microbenchmark
suite that measures the performance of kernel operations (system calls).
We run the LEBench measurement for each system call \repeatnum times and report the average latency overhead in the top half of Figure~\ref{fig:lebench-overhead}, comparing \sys and the insecure baseline.
\sys introduces an average overhead \lebenchoverheadavg across all system calls in the suite, which is almost negligible.

However, unlike the SPEC benchmark, we observe large performance variations in the LEBench, ranging from \lebenchoverheadmin to \lebenchoverheadmax.
To understand the variation, for each system call, we plot the range of multiple measured latencies normalized to the medium latency of the insecure baseline in the bottom half of Figure~\ref{fig:lebench-overhead}.
Specifically, the colored box indicates first quartile (Q1) and third quartile (Q3) of normalized latency,
with the medium latency marked as a horizontal line dividing the box into two halves.
We additionally use hollow dots to mark outliers.

From the figure, we observe that large variations consistently exist for certain system calls, such as \texttt{thrcreate} and \texttt{pagefault}, in the baseline and when using \sys. Given that these system calls require coordination between multiple threads, we suspect that the large variation is caused by the dynamic non-deterministic scheduling of the Linux kernel.

Another difference between LEBench and SPEC2017 is the ratio of memory accesses that use \sys's protection.
In the LEBench benchmark, the average ratio of memory accesses that need $\addrmask$ is \lebenchmaskratioavg, ranging from \lebenchmaskratiomin to \lebenchmaskratiomax, which is lower than the ratio in SPEC2017.
This is because LEBench contains more memory accesses to kernel data, which is not protected by our prototype implementation.

\begin{figure}[t]
    \centering
    \includegraphics[width=\linewidth,trim={0 0.4cm 0 0.4cm},clip]{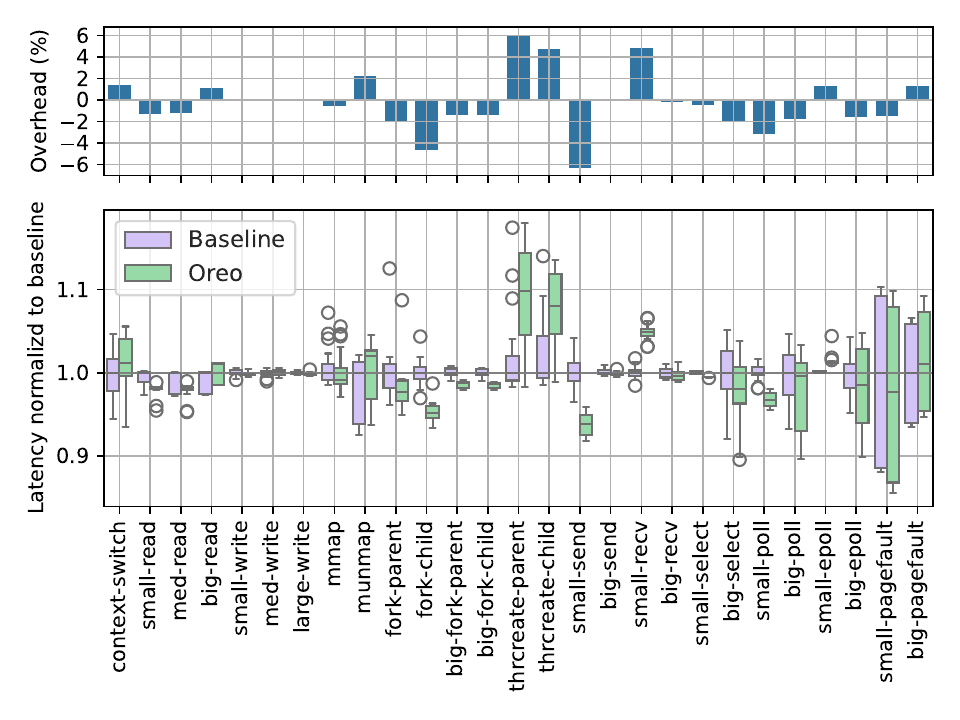}
    \caption{Performance evaluation results on LEBench}
    \label{fig:lebench-overhead}
\end{figure}

\subsection{Security Evaluation}
\label{sec:security-evaluation}
We conduct three experiments to validate the security properties of \sys and to demonstrate that our implementation aligns with the design presented in the paper.
\ifextendedversion
Additionally, we provide a formal proof in Appendix~\ref{sec:proof} to show that \sys achieves a non-interference property, preventing the leakage of ASLR secrets.
\else
Additionally, we provide a formal proof in a technical report~\cite{song-oreo-proof} to show that \sys achieves a non-interference property, preventing the leakage of ASLR secrets.
\fi

\pgheading{The Prefetch Attack}
We evaluate a prefetch attack~\cite{gruss2016prefetch,liu2023entrybleed,lipp2022amd} on the insecure baseline and \sys.
In both cases, the kernel code (including text and modules) is randomly relocated to the range from \texttt{0xffffff8000000000} to \texttt{0xffffffef00000000}.
The attacker scans the randomization region by probing addresses with a stride of $\SI{2}{GB}$.
The probing operation executes the prefetch instruction twice.
The first fetch operation brings the address into various microarchitecture structures.
The attacker then measures the latency of the second fetch operation.

Figure~\ref{fig:prefetch-plot} shows the attack results.
The prefetch attack works effectively on the insecure baseline, where the prefetch latency is distinctively lower at the randomized kernel address \texttt{0xffffff8601800040}, and is consistently high at the other unmapped kernel addresses.
This timing difference is caused by the fact that TLB caches address translation only for mapped addresses, not for unmapped addresses.
In contrast, the prefetch attack no longer works on \sys.
\sys converts virtual addresses to masked addresses and uses masked addresses to access the TLB.
Since the TLB already cached the corresponding masked address during the execution of the first prefetch operation, the second prefetch always results in a hit, resulting in indistinguishable low latency.

\begin{figure}[t]
    \centering
    \includegraphics[width=\linewidth]{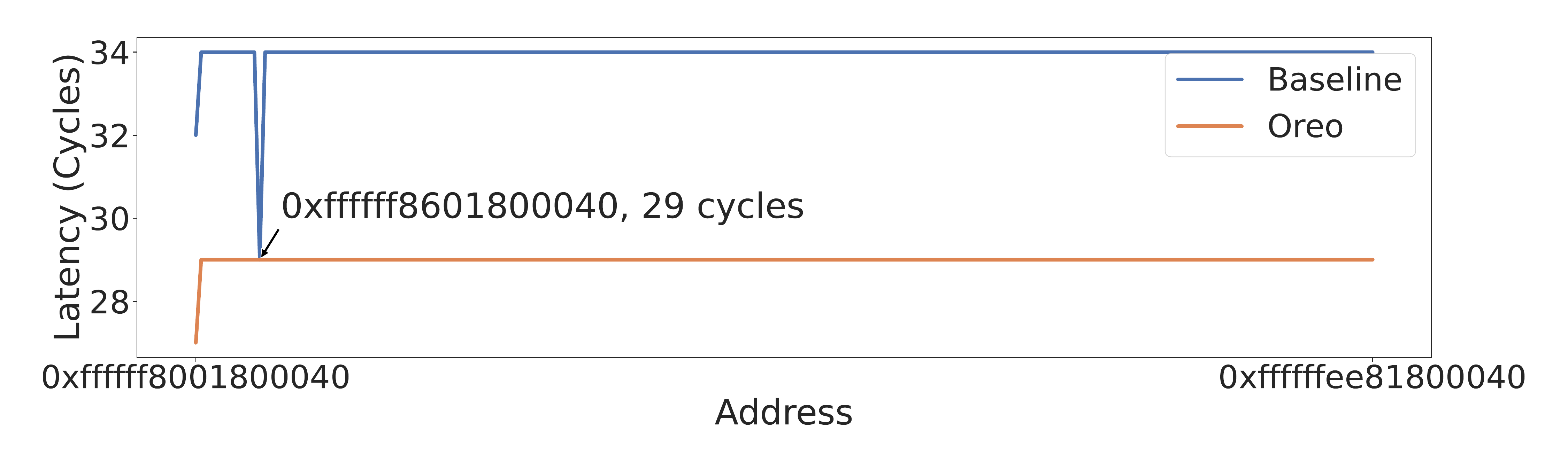}
    \caption{Evaluating the prefetch attack on the insecure baseline and \sys.}
    \label{fig:prefetch-plot}
\end{figure}

\begin{figure*}[t]
    \centering
    \includegraphics[width=\linewidth]{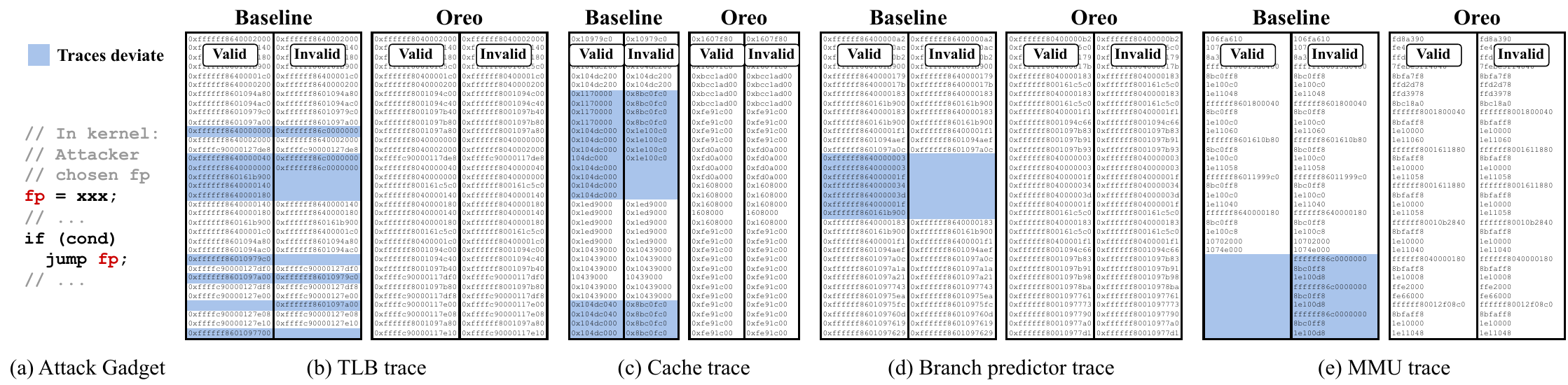}
    \caption{Evaluating the speculative code gadget probing attack on the insecure baseline and \sys. (b)-(e) compare the input traces of multiple microarchitecture structures when transiently jumping to a valid and an invalid address.}
    \label{fig:sec-trace}
\end{figure*}

\pgheading{Leakage Path \circled{1}}
In addition to the prefetch attack, we evaluate the speculative code region probing attack from the BlindSide paper~\cite{goktas2020speculative}.
We consider this attack as a representative attack for the leakage path \circled{1}.
Besides, we think solely relying on timing may not fully capture the effectiveness of our design.
Thus, we validate the attack results by comparing internal microarchitectural traces dumped from the gem5 simulator.

Figure~\ref{fig:sec-trace}(a) shows the attack gadget.
Following the attack described in BlindSide~\cite{goktas2020speculative}, we assume an attacker leverages a memory corruption vulnerability to overwrite a function pointer in the kernel memory.
The attacker then triggers the victim to transiently jump to the corrupted function pointer during the mis-speculation of a conditional branch.
In our experiment, we set the function pointer to a valid address and an invalid address and compare the input traces on TLB, cache, branch predictor, and MMU, shown in Figure~\ref{fig:sec-trace}(b)-(e).
We highlight the parts where the two traces deviate from each other, indicating an exploitable side channel.

On the baseline, the input traces to the four microarchitecture structures all deviate when transiently accessing the valid and the invalid addresses. 
In contrast, with \sys, the traces are identical since the same masked address is used for transient access.
This indicates the attacker cannot distinguish valid addresses from invalid ones through transient memory accesses, so \sys can successfully block leakage path \circled{1}.

\pgheading{Leakage Path \circled{2}}
We use microarchitectural traces, similar to Figure~\ref{fig:sec-trace}, to evaluate \sys against the second leakage path.
We trigger the victim kernel to execute a system call and record the input traces to various microarchitecture structures. 
We then examine these traces and find that on the baseline, the input traces to these microarchitecture structures all include addresses with secret $\oreooffset$. %  random offsets.
With \sys, the input traces only consist of addresses without these secret bits.

%% file: tex/related.tex
\section{Related Work}
\label{sec:related}

We discuss related work aimed at strengthening the security property of ASLR schemes.
We first discuss ASLR protection schemes that aim to address different information leakage threats, i.e., through microarchitectural attacks or software-level attacks.
We then discuss mechanisms designed to increase ASLR entropy.

\pgheading{Blocking Software-Based ASLR Bypasses}
Several prior work~\cite{lu2015aslr,vano2020kaslr,vano2020info,backes2014you,crane2015readactor,crane2015s,gionta2016preventing} aim to make it more difficult for attackers to leak ASLR secrets via exploiting software vulnerabilities.
For example, ASLR-Guard~\cite{lu2015aslr} uses encryption to prevent leaking code pointers.
KASLR-MT~\cite{vano2020kaslr} and Vano-Garcia et al.~\cite{vano2020info} block information leakage due to memory deduplication attacks.
XnR~\cite{backes2014you} and Readactor~\cite{crane2015readactor,crane2015s}, and Gionta et al.~\cite{gionta2016preventing} enable ``executing-only-memory'' to prevent reading and then leaking code pointers.
These techniques focus on software-level threats and are ineffective towards microarchitectural attacks.
They can complement \sys to further strengthen the security of ASLR schemes.

\pgheading{Blocking Microarchitectural-Attack-Assisted ASLR Bypasses}
This group of mitigation mechanisms~\cite{gruss2017kaslr,canella2020kaslr,gens2017lazarus} shares the same goal as our work.
Several defenses aim to prevent ASLR secrets from being leaked via virtual memory layout probing attacks, blocking leakage path \circled{1}.
One approach is to isolate the kernel and user-space address spaces to prevent memory layout probing, as demonstrated in
KAISER~\cite{gruss2017kaslr} (also known as KPTI) 
and LAZARUS~\cite{gens2017lazarus}.
However, given that some of the kernel trampoline pages are still mapped in the user space, the Linux prototype of KAISER~\cite{linux_kpti} is still vulnerable to ASLR bypasses, as shown in EchoLoad~\cite{gruss2017kaslr} and Entrybleed~\cite{liu2023entrybleed}.
Besides, software-level isolation does not help mitigate the second leakage path, where the ASLR randomized bits are leaked when the victim program uses secret pointers as program counters or load/store addresses.

Alternatively, 
FLARE~\cite{canella2020kaslr} makes accessing a valid (mapped) and invalid (unmapped) kernel addresses take the same amount of time.
It works by mapping all invalid kernel addresses to one valid physical page, so that accessing these invalid pages will need to go through the full page table walk.
This mitigation has several limitations.
First, it cannot block memory layout probing attacks that use BTB and TLB as side channels.
Second, given that the invalid addresses now map to an empty physical page, whose content is different from the actual valid pages, any speculative data-dependent accesses can be used to distinguish between invalid and valid addresses.
In contrast to this ad-hoc solution, \sys takes a much more comprehensive approach to make accessing invalid and valid addresses exhibit indistinguishable side effects as long as they map to the same masked address.

\pgheading{Increasing ASLR Entropy}
Several prior works aim to increase the entropy of ASLR~\cite{gallagher2019morpheus,nikolaev2022adelie,crane2015readactor,chen2017codearmor,accardi2020finer}
by introducing finer-grained ASLR, increasing the size of the randomization region, and re-randomization.
For example, FGKASLR~\cite{accardi2020finer} randomizes function orders in kernel code to make it harder to find code gadgets.
Adelie~\cite{nikolaev2022adelie} increases the kernel ASLR randomization region to the entire 64-bit virtual memory and re-randomizes the layout of kernel modules. 
Readactor~\cite{crane2015readactor,crane2015s} and CodeArmor~\cite{chen2017codearmor} re-randomizes code pointers or code memory layout. 
Morpheus~\cite{gallagher2019morpheus} periodically re-randomizes the code pointers with a higher frequency than previous approaches as it leverages complex hardware modifications to do the re-randomization.

\pgheading{Power-Induced Timing Side Channels}
Recent work~\cite{lipp2022amd} breaks ASLR through power-induced timing side channels.
Although it is out of \sys's protection scope,
\sys's implementation can make it easier to mitigate this class of side channels.
Power-based side channels arise 
when circuits perform computation using secret-dependent values.
Different inputs activate different transistors and lead to
different amounts of power consumption.
According to Figure~\ref{fig:design-uarch}, the protected secret bits are extracted from virtual addresses or obtained from TLB entries, and are only used in two equality checks (denoted by the ``\texttt{=?}'' boxes).
Thus, it becomes feasible to adopt classic mitigations, such as circuit masking~\cite{mangard2008power}, to secure these specific modules against power side channels.

%% file: tex/conclusion.tex
\section{Conclusion}
This paper presented a systematic analysis of microarchitectural side-channel attacks for ASLR bypasses.
We use our analysis to guide the design of \sys.
\sys strengthens the security of ASLR via the modification to the memory interface.
By introducing an extra layer of masked address space and converting virtual addresses to masked addresses to be used to set up page tables and index into various microarchitecture structures, \sys constraints the secret offset exposure in both software and hardware.
Our security and performance evaluation demonstrates that \sys is an effective mitigation and introduces negligible performance overhead.

%% file: tex/acknowledgment.tex
\section*{Acknowledgment}
The authors thank the Matcha Group (MIT) for their help and the anonymous NDSS reviewers for their feedback.
This work was supported in part by a gift from Amazon; by the Air Force Office of Scientific Research (AFOSR) under grant FA9550-22-1-0511;  by ACE, one of the seven centers in JUMP 2.0, a Semiconductor Research Corporation (SRC) program sponsored by DARPA.


%% file: tex/artifact.tex
% Artifact Appendix template for the NDSS Artifact Evaluation
% version 1.0 (20230620)

% remove the following block when merging the appendix with the camera-ready full paper
%%%
% \documentclass[conference]{IEEEtran}
% \input{tex/top}
% \pagestyle{plain}
% \usepackage{url}
% \usepackage[caption=false, font=footnotesize]{subfig}

% \newcommand{\BashFancyFormatLine}{%
%   \def\FancyVerbFormatLine##1{\$\ \,##1}%
% }
% \setminted[bash]{
%     linenos=false,
%     % breaksymbolleft={\quad},
%     % breaksymbolright={\texttt{\small\textbackslash}},
%     % breaksymbolsepright=0pt,
%     % breaksymbolindentright=2pt,
%     % breaksymbolsepright=0pt,
%     formatcom=\BashFancyFormatLine,
% }
% \let\labelindent\relax
% \usepackage{enumitem}

% \begin{document}
%%%

% \appendices

\section{Artifact Appendix}
\subsection{Description \& Requirements}

\subsubsection{How to access}
The artifact for reproducing the results in the paper is at 
\url{https://doi.org/10.5281/zenodo.14261065}. The source code of our implementation and experiments is available at \url{https://github.com/CSAIL-Arch-Sec/Oreo}.

\subsubsection{Hardware dependencies}
The artifact utilizes the gem5 simulator (v24.0)~\cite{Binkert:2011:gem5,Lowe-Power:2020:gem5-20} to emulate Linux and execute benchmarks in full-system mode, requiring only CPU, memory, and disk resources.
Each gem5 instance needs 1 core and $\SI{2}{GB}$ memory. 
The artifact's resources, including the source files and a generated disk image, require $\SI{50}{GB}$ of available disk space.
We estimate the runtime of each experiment in this artifact on an Intel(R) Xeon(R) Gold 5220R CPU with $\SI{252}{GB}$ memory.

\subsubsection{Software dependencies}
The artifact requires Vagrant (v2.4) and VirtualBox (v7.0) to build a disk image for running experiments in gem5 full-system mode.
It also requires Docker Engine (v24+) and Docker Compose (v1.29+) to run all the experiments.

\subsubsection{Benchmarks}
Reproducing the performance evaluation results requires SPEC 2017 benchmark~\cite{bucek2018spec} and LEBench benchmark~\cite{zzrcxb-LEBench-Sim}.
LEBench is included in the artifact source, while SPEC 2017 is not included because of copyright issues.

\subsection{Artifact Installation \& Configuration}
\subsubsection{Pre-setup}
In the artifact folder, run the following command to clone the git submodules:
\begin{minted}{bash}
git submodule update --init --recursive
\end{minted}
Note that \path{cpu2017-1.1.9.iso} is the official image for SPEC 2017 benchmark, which is not released with the artifact since it is not open-source. It should be placed in the artifact folder.
The file structure is as shown in Figure~\ref{fig:struct-outside}.

\begin{figure}[!h]
    \centering
\begin{tikzpicture}%
    \draw[color=black!60!white]
    \FTdir(\FTroot,artifact,artifact){
        \FTdir(artifact,experiments,experiments~~~~~\# Sources files to build disk image)
        \FTdir(artifact,gem5,gem5~~~~~~~~~~~~\# \sys's microarchitectural changes)
        \FTdir(artifact,linux,linux~~~~~~~~~~~\# \sys's kernel changes)
        \FTfile(artifact,cpu2017-1.1.9.iso~\;\# SPEC 2017 official image)
    };
\end{tikzpicture}
    \caption{Artifact file structure}
    \label{fig:struct-outside}
\end{figure}
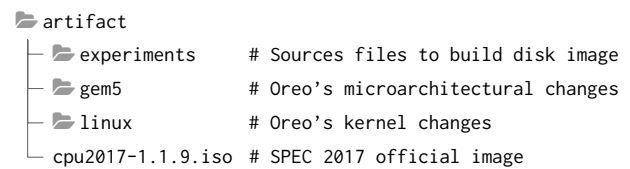

\subsubsection{Build}
In folder \path{artifact}, set up and enter the docker by running:
\begin{minted}{bash}
cd linux
docker-compose up -d
docker-compose exec x86_fs /bin/bash
\end{minted}
In the docker's shell, build Linux and gem5 by running:
\begin{minted}{bash}
cd /root/linux
python3 compile_scripts/compile.py --num-cores=80
cd /root/gem5
python3 scripts/compile.py --num-cores=80
\end{minted}
Please replace \path{80} with an appropriate number of CPU cores used for the build based on CPU and memory resources.

In the host's shell and in folder \path{artifact}, run the following commands to build the disk image, which takes about 1 hour.
\begin{minted}{bash}
tar -cvf linux.tar linux
cd experiments/disk-image
vagrant up
vagrant halt
vagrant global-status --prune
ls ~/VirtualBox\ VMs
# Find vm directory denoted as <vm-directory>
qemu-img convert -f vmdk -O raw ~/VirtualBox\ VMs/<vm-directory>/ubuntu-jammy-22.04-cloudimg.vmdk experiments.img
\end{minted}

The \path{artifact} folder is mounted at \path{/root} in the docker. In the following document, we will refer to the path inside the docker by default.

\subsection{Experiment Workflow}
For each experiment, we use gem5 to boot Linux under the full system mode and run benchmarks stored in the generated disk image \path{/root/experiments/disk-image/experiments.img}.
We provide scripts under the \path{/root/gem5/scripts} directory to conveniently start experiments using different setups (baseline and \sys) and benchmarks (e.g., SPEC 2017, LEBench and ASLR attacks).

\subsection{Major Claims}
\begin{itemize}
    \item (C1): \sys introduces negligible overhead compared to the insecure baseline, as demonstrated by experiments E2 and E3 in Figure~\ref{fig:spec}, \ref{fig:lebench-overhead}.
    \item (C2): In SPEC 2017 and LEBench, the ratio of memory accesses where \sys applies its protection (i.e., the ratio of masked memory accesses) is high, which indicates that the benchmarks are relevant to evaluating the performance impacts of \sys. This is demonstrated by experiment E5.
    \item (C3): \sys prevents the leakage of the ASLR secret offset through paths \circled{1} and \circled{2}, whereas the insecure baseline does not. This is demonstrated by experiment E4, as shown in Figure~\ref{fig:prefetch-plot}, \ref{fig:sec-trace}.
\end{itemize}

\subsection{Evaluation}
\subsubsection{Experiment (E1)}
[Functional Test] [1 human-minute + 10 compute-minutes]:
This experiment runs two user programs \path{hello} and \path{hello_invalid}, whose source code can be found in \path{/root/experiments/disk-image/experiments/experiments/src}.
This experiment demonstrates that \sys allows valid programs (e.g. \path{hello}) to run successfully while raising exceptions on invalid memory accesses at commit time (e.g. \path{hello_invalid}).

\textit{[Preparation]}
Enter the docker and change the working directory:
\begin{minted}{bash}
# Host, in folder artifact/linux
docker-compose exec x86_fs /bin/bash
# Docker
cd /root/gem5
\end{minted}
We will use \path{/root/gem5} as the default working directory for running the following experiments and identifying the path to output files.

\textit{[Execution]}
Running the following command, which takes 10 minutes:
\begin{minted}{bash}
python3 scripts/gen_checkpoint.py && python3 scripts/run_example.py
\end{minted}

\textit{[Results]}
\path{hello} is a valid program, so gem5 should finish simulation without any exceptions. Check \path{result/restore_ko_111_0c0c00/hello_/board.pc.com_1.device} for its output.
\path{hello_invalid} is a malicious program that accesses an invalid address. In the directory \path{result/restore_ko_111_0c0c00/hello_invalid_}, check \path{board.pc.com_1.device} for its output; check \path{stderr.log} for the gem5 exception message on committing invalid memory accesses.
The exception message can be located by searching for ``ASLR violation'' in the log file.

\subsubsection{Experiment (E2)}
[Performance Evaluation with SPEC 2017]
% We are supposed to add performance overhead measurement on SPEC 2017.
This experiment evaluates \sys's performance against the baseline using SPEC 2017 Intrate Benchmarks.
In our original setup, we warm up the system and microarchitecture structures by executing 10 billion user-space instructions, and then measure the performance of the next 1 billion user-space instructions. This setup takes several days to finish. For the artifact evaluation, we provide a scaled-down option to warm up with 1 billion user instructions. 
\sys introduces negligible overhead compared to the baseline.

\textit{[Preparation]} Same to E1

\textit{[Execution]} Run the following command:
\begin{minted}{bash}
python3 scripts/run_spec.py --gen-cpt --begin-cpt=1 --num-cpt=1 --num-cores=80 --user-delta=32 --spec-inst-warmup-step=1
\end{minted}

\textit{[Results]} Run the following command to parse the result:
\begin{minted}{bash}
python3 scripts/parse_spec.py --parse-raw --begin-cpt=1 --num-cpt=1 --roi-idx=2 --expected-stats=3
python3 scripts/parse_spec.py --begin-cpt=1 --num-cpt=1 --roi-idx=2 --expected-stats=3
\end{minted}
The CPI overhead introduced by \sys is recorded in \path{scripts/spec_output/merge_input_mean_user_cpi_1_2.pdf}.
The CPI overhead is expected to be negligible, as shown in Figure~\ref{fig:spec}. However, these benchmarks behave differently in different regions of interest. Hence, the measured absolute CPI might differ from the CPI in Figure~\ref{fig:spec} when using 1 billion warmup instructions.

\subsubsection{Experiment (E3)}
[Performance Evaluation] [1 human-minute + 8 compute-hours]:
This experiment evaluates \sys's performance against the baseline using LEBench, which measures and reports the average latency of various system calls. We ran the benchmark multiple times and used the average latency across all iterations as the final measurement. As shown in Figure 7, \sys introduces negligible average overhead compared to the baseline. The observed variance in latency across different runs is due to the dynamic, non-deterministic scheduling of the Linux kernel, which explains the performance gap between \sys and the baseline in some benchmarks.

\textit{[Preparation]} Same to E1.

\textit{[Execution]} 
We repeated the measurement 16 times to minimize the impact of dynamic scheduling on latency measurement. For artifact evaluation, we offer a scaled-down option to repeat the measurement 8 times by running:
\begin{minted}{bash}
python3 scripts/run_perf.py --gen-cpt --num-cpt=8 --begin-cpt=0 --num-cores=80
\end{minted}
To modify the number of repetitions, please adjust the \path{--num-cpt} option.
This process takes 8 hours to complete.

\textit{[Results]}
Run the following commands to parse experiment results and generate Figure~7:
\begin{minted}{bash}
python3 scripts/parse_perf.py --suffix-range=0,8 --plot --parse
\end{minted}
The output files are stored in \path{/root/gem5/scripts/plot}. 
Due to Linux's dynamic scheduling, the overhead measurements may vary from those shown in Figure 7; however, the average overhead is expected to be negligible.

\subsubsection{Experiment (E4)}
[Security Evaluation] [5 human-minute + 10 compute-minutes]:
This experiment consists of three parts:
\begin{itemize}
    \item The Prefetch Attack.
    \item Leakage path \circled{1}: 
    we use the speculative code region attack as a representative example of the leakage path \circled{1}. We demonstrate that speculatively accessing both valid and invalid addresses produces different microarchitectural traces on the baseline, leading to the leakage of the secret offset. In contrast, \sys produces the same effects regardless of the address validity, confirming that \sys effectively blocks the leakage path \circled{1}.
    \item Leakage path \circled{2}: we trigger a system call and record input traces to various microarchitecture structures, which indicates whether the system leaks secrets from the leakage path \circled{2}.
\end{itemize}

\textit{[Preparation]} Same to E1.

\textit{[Execution]}
Run the following command to run the three parts, which takes about 10 minutes.
\begin{minted}{bash}
python3 scripts/gen_checkpoint.py && python3 scripts/run_sec.py
\end{minted}

\textit{[Results]}
\begin{itemize}
    \item The prefetch attack: \path{scripts/parse_prefetch.py} is used to generate Figure~8 at \path{scripts/plot/prefetch_plot.pdf}, which indicates that \sys prevents the prefetch attack from leaking secrets through the timing side channel.
    \item Leakage path \circled{1}: For the baseline, the microarchitecture traces of speculatively accessing valid and invalid addresses are at \path{result/restore_ko_000_0c0c00/blindside_1_0c_/trace.out.gz} and \path{result/restore_ko_000_0c0c00/blindside_1_0d_/trace.out.gz}. For \sys, the traces are \path{result/restore_ko_111_0c0c00/blindside_1_0c_/trace.out.gz} and \path{result/restore_ko_111_0c0c00/blindside_1_0d_/trace.out.gz}. Traces of the baseline are not identical, while traces of \sys are identical, which verifies that \sys blocks the first leakage path.
    \path{scripts/parse_trace.py} is used to extract the condensed trace demonstrated in Figure~9 in the directory \path{scripts/plot/trace} and compare them.
    \item Leakage path \circled{2}: check \path{iTLBWalker} and \path{dTLBWalker} traces for both setups (generated by \path{scripts/parse_trace.py}), which prints the virtual addresses that need to be translated and physical addresses accessed during page table walk. On the baseline, secret dependent physical addresses are used for page table walk, while on \sys, secret independent addresses are used. Hence, \sys blocks leakage path \circled{2}.
\end{itemize}

\subsubsection{Experiment (E5)}
This experiment reports the ratio of memory accesses where \sys applies its protection (i.e., the masked ratio).
This is used to evaluate whether the benchmark is relevant to measure the performance impact of \sys. 

\textit{[Results]}
The masked ratio of SPEC 2017 benchmark can be found in the last column of \path{scripts/spec_output/merge_input_oreo_user_cpi_1_2.csv}.
Run the following command, and the masked ratio of LEBench can be found in the last column of \path{scripts/lebench_output/test_mask_ratio.csv}.
\begin{minted}{bash}
python3 scripts/parse_perf_stats.py --begin-cpt=0 --num-cpt=8
\end{minted}

% --------------------------------------------------------

% remove the following block when merging the appendix with the camera-ready full paper
%%%
% \end{document}
%%%

%% file: tex/appendix.tex
\section{Attack Summary}
\label{sec:full-attack-summary}
To complement the attack analysis in Section~\ref{sec:attack-analysis}, we provide a detailed enumeration of existing microarchitectural-attack-assisted ASLR bypasses, including the leakage paths they take and the utilized side channels in Table~\ref{tab:full-attack-summary}.
Blindside~\cite{goktas2020speculative} involves several attacks including Code Region Probing and Spectre Probing.
We demonstrated that \sys prevents Code Region Probing in our security evaluation (Section~\ref{sec:security-evaluation}).
Spectre Probing is mitigated by existing Spectre mitigations~\cite{yu2019speculative,weisse2019nda,yan2018invisispec, ainsworth2021ghostminion,choudhary2021speculative,barber2019specshield,fustos2019spectreguard,koruyeh2020speccfi,schwarz2020context,yu2018data,daniel2023prospect,loughlin2021dolma,mosier2023serberus,yu2020speculative,ainsworth2020muontrap,khasawneh2019safespec,kiriansky2018dawg,li2019conditional,saileshwar2019cleanupspec,sakalis2019ghost,sakalis2019efficient}.
The Prefetch+Power attack in \cite{lipp2022amd} utilizes power side channels, which are not included in our threat model.
\sys successfully mitigates all other attacks in Table~\ref{tab:full-attack-summary}.

\begin{table*}
\centering
\caption{Summary of microarchitectural-attack-assisted ASLR bypasses. The numbers in the second column refer to the three leakage paths in Figure~\ref{fig:aslr-attack-graph}.}
\small
\begin{tblr}{
    colspec={Q[c]|Q[c]|Q[l]},
    rowspec={Q[m]},
    row{1} = {font=\bfseries},
    rows = {belowsep=1pt,abovesep=3pt},
    % hline{1-Z} = {1}{-}{},
    hline{1, 2, 3, 5, 6, 7, 9, 10, 11, 13, 15, 17, 19, 20, 21} = {1}{-}{},
    hline{2} = {2}{-}{},
    % hline{1,2,11,Y,Z} = {1pt},
    vline{1,4} = {abovepos = 1, belowpos = 1},
}
    \SetCell{c} {Leakage\\Path} & \SetCell{c} Attacks &  Side Channels \\
    \SetCell[r=9]{c}\circled{1} & Code Region Probing~\cite{goktas2020speculative} &  ICache Prime+Probe \\ 
    & Double Page Fault~\cite{hund2013practical} &  \SetCell[r=2]{l} Page fault latency to check TLB states \\
    &  DrK~\cite{jang2016breaking} & \\
    & Osiris~\cite{weber2021osiris} &  
    Cache side channel to monitor pipeline squashes caused by page fault
    \\
    & EchoLoad~\cite{canella2020kaslr} &  Cache side channel to monitor whether speculation is stalled \\
    & Data Bounce~\cite{schwarz2019store} &  \SetCell[r=2]{l} Cache side channel to monitor store-to-load forwarding behaviors \\
    & Fallout~\cite{canella2019fallout} &  \\
    & AMD Prefetch+Time~\cite{lipp2022amd} &  Prefetch latency to monitor page table walks \\
    & AMD Prefetch+Power~\cite{lipp2022amd} &  Prefetch power consumption to monitor page table walks \\
    \SetCell[r=9]{c}\circled{2} & Gruss et al.~\cite{gruss2016prefetch}  & \SetCell[r=2]{l} Prefetch latency to check TLB states \\
    & EntryBleed~\cite{liu2023entrybleed} &  \\
    & TLBleed~\cite{gras2018translation} &  \SetCell[r=2]{l} TLB Prime+Probe \\
    & TagBleed~\cite{koschel2020tagbleed} &  \\
    & Jump Over ASLR~\cite{evtyushkin2016jump} &  \SetCell[r=2]{l} BTB Prime+Probe \\
    & Phantom~\cite{wikner2023phantom} &  \\
    & AnC~\cite{gras2017aslr} &  \SetCell[r=2]{l} Cache Prime+Probe to monitor page table walks \\
    & Binoculars~\cite{zhao2022binoculars} &  \\
    & Take A Way~\cite{lipp2020take} &  DCache Collide+Probe \\
    \circled{3} & Spectre Probing~\cite{goktas2020speculative} &  ICache or DCache Prime+Probe \\
\end{tblr}
\label{tab:full-attack-summary}
\end{table*}

%% file: tex/proof.tex
\section{Security Proof}
\label{sec:proof}
We provide formal proof to show that \sys achieves a non-interference property to prevent attackers from distinguishing virtual memory layouts with different secret offsets.

\subsection{Abstract Machines With Memory Mapping Interfaces}
Prior works~\cite{guarnieri2021hardware,mosier2023serberus} define operational semantics for processors, however, without modeling the memory mapping interface.
In our security proof, we first define our own operational semantics for an abstract machine that incorporates an abstract memory mapping interface.

We model an abstract machine consisting of a physical memory and a series of microarchitecture states.
We denote the machine state as $S\!\coloneqq\! \qangle{m, \mu}$, where $m$ refers to the physical memory, and $\mu$ refers to the state of other hardware components.

\pgheading{Memory Configuration}
The physical memory is split into three parts: program, data, and page table structures, denoted as $m\!\coloneqq\! \qangle{P, D, PT}$.
The program $P$ is a map from physical addresses to instructions, denoted as $P:[\pistart, \piend)\rightarrow Inst$, denoting how instructions are stored in physical memory.
Similarly, The page table $PT$ is a mapping function to describe how page table entries (PTEs) are stored in memory, and the data memory $D$ describes other data.
Both $P$ and $PT$ are read-only after initialization.

\pgheading{Modeling Memory Mapping Interface}
We define the memory layout $L$ as a function that maps each virtual address to either a physical address or nothing ($\bot$).
If the program specifies accessing a virtual address $v$, the machine should access address $L(v)$ in the physical memory if $v$ is valid, or eventually crash if invalid.
We model ASLR as randomly selecting a layout from the set $\mathcal{L}\!\coloneqq\!\qty{L_i: i \!=\! 0, 1, \dots, N}$ and initializing the machine state accordingly so that the code pointers are randomized.

We define two functions to model the address translation procedure.
$\trans(x, PT)$ returns the address translation result for an address $x$ using the information in page table $PT$.
The address $x$ is a virtual address in the baseline and a masked address in \sys.
It returns $\bot$ if $x$ is invalid and otherwise returns the mapped physical address.
Specifically, in the baseline, $\trans(x, PT)\!=\!L(x)$.

$\ptw(x, PT)$ returns a sequence of physical addresses of page table entries needed for translating the given address $x$ to its corresponding physical address.
The sequence of physical addresses will be wrapped as memory requests to interact with various hardware components in the abstract machine.

\pgheading{Microarchitecture State}
We define the microarchitecture state as $\mu \!\coloneqq\! \qangle{BP, LSQ, Cache, TLB, E}$.
$BP$ is the branch prediction unit, updated by program counters (PCs).
$LSQ$ records load/store addresses, but does not record the load/store data.
Similarly, $Cache$ records the addresses of memory blocks that reside in the cache (i.e., the metadata), but does not record the data in cache blocks. 
It is updated by memory accesses issued by instruction fetch, load/store instructions, and page table walk operations.
$TLB$ is updated by a memory access's virtual address and its address translation result.
$E$ represents all other components inside the core, including the scheduler, reorder buffer, and register files.
These structures are not indexed by virtual or physical addresses and, thus, are less relevant to ASLR security.

We model four types of requests generated by $E$.
$\progreq(E)\!=\!\noneop() | \fetchop(v, v_\text{src}) | \loadop(v) | \storeop(v, d)$, where $v$ is the target virtual address of the corresponding fetch, load, or store operation and $d$ is data written to memory by the store operation.
The fetch operation additionally includes a $v_\text{src}$ argument.
When $v$ is the target address of a branch, $v_\text{src}$ is the source address of that branch. 
We use the branch source and destination information to update branch predictors.

This abstract model covers a variety of pipeline designs. 
It can model an out-of-order processor that supports speculative execution.
For example, $E$ takes input from the branch predictor, sends speculative $\fetchop$ requests to the memory system, and handles squashes on incorrect speculative execution internally.
Similarly, both $\loadop$ and $\storeop$ requests can be speculative. When $E$ finishes calculating load/store addresses of instructions in the reorder buffer, it sends the requests to $LSQ$, $TLB$, and $Cache$.
Furthermore, $E$ takes input from the physical memory $p$ and other microarchitecture components in $\mu$ to update its state.

\pgheading{Adversary Model}
We define the adversary's observation function over a state of the abstract machine $S$ as $\uarchobs(S)\!=\!\qangle{BP, LSQ, Cache, TLB}$.
This function defines a strong adversary model where the attacker monitors all microarchitecture structures indexed by addresses.

\subsection{Initialization and Execution Semantics}
Given a program $P$ and a virtual memory layout $L$, we initialize the hardware state $S\!=\!\init(P, L) \!=\! \qangle{P, D_L, PT_L, BP_0, LSQ_0, Cache_0, TLB_0, E_L}$. 
The page table $PT$,  the data memory $D$, and core components $E$ need to be initialized using the layout information $L$.
Other structures are initialized using a reset empty state.

We formalize hardware semantics to model the execution of the abstract machine.
We denote one-step execution as $S\xrightarrow{\progreq(E)}S'$, where functional structures in $E$ of the current state $S$ generate requests $\progreq(E)$, and the whole microarchitecture state is updated to $S'$.

A $t$-step execution of the machine is denoted as $S_0\rightarrow^*S_t$, representing $S_0 \xrightarrow{\progreq(E_0)} S_1 \dots \xrightarrow{\progreq(E_{t-1})} S_t$.
We denote the request traces as $\progreq(S_0, t)\!=\!\progreq(E_0),\progreq(E_1),\dots,\progreq(E_{t-1})$.
Correspondingly, we define the adversary's observation of the machine as a trace $\uarchobs(S_0, t) \!=\! \uarchobs(S_0), \uarchobs(S_1), \dots, \uarchobs(S_t)$.

The machine proceeds one step only if the request does not cause a crash. We prove \sys's security properties hold for the case when the machine does not crash.

\pgheading{Baseline Execution Semantics}
Given a page table $PT$ initialized using a layout $L$, we denote the page table as $PT_L$.
For all virtual addresses $v$, the hardware translation results match the layout query results, i.e., $\trans(v, PT_L) \!=\! L(v)$.
We use $\update(X, args)$ to denote how the microarchitecture structure $X$ is updated with input arguments $args$. For convenience, $args$ only lists the inputs that affect the state of $X$.
\begin{mathpar}
\small
\inferrule[Fetch]{
\progreq(E)\!=\!\fetchop(v, v_\text{src})\\
BP' \!=\! \update(BP, v, v_\text{src}) \\
Cache' \!=\! \update(Cache, \ptw(v, PT), \trans(v, PT), TLB)\\
TLB' \!=\! \update(TLB, v, \trans(v, PT))\\
E' \!=\! \update(E, BP', LSQ, Cache', TLB', P[\trans(v, PT)])
}{
\qangle{P, D, PT, BP, LSQ, Cache, TLB, E} \xrightarrow{\fetchop(v, v_\text{src})}\\
\qangle{P, D, PT, BP', LSQ, Cache', TLB', E'}
}
\end{mathpar}
\begin{mathpar}
\small
\inferrule[Load]{
\progreq(E)\!=\!\loadop(v)\\
LSQ' \!=\! \update(LSQ, v, \trans(v, PT)) \\
Cache' \!=\! \update(Cache, \ptw(v, PT), \trans(v, PT), TLB)\\
TLB' \!=\! \update(TLB, v, \trans(v, PT))\\
E' \!=\! \update(E, BP, LSQ', Cache', TLB', D[trans(v, PT)])
}{
\qangle{P, D, PT, BP, LSQ, Cache, TLB, E} \xrightarrow{\loadop(v)}\\
\qangle{P, D, PT, BP, LSQ', Cache', TLB', E'}
}
\end{mathpar}
\begin{mathpar}
\small
\inferrule[Store]{
\progreq(E)\!=\!\storeop(v,d)\\
D'\!=\!\update(D, \trans(v, PT), d)\\
LSQ' \!=\! \update(LSQ, v, \trans(v, PT)) \\
Cache' \!=\! \update(Cache, \ptw(v, PT), \trans(v, PT), TLB)\\
TLB' \!=\! \update(TLB, v, \trans(v, PT))\\
E' \!=\! \update(E, BP, LSQ', Cache', TLB')
}{
\qangle{P, D, PT, BP, LSQ, Cache, TLB, E} \xrightarrow{\storeop(v,d)}\\
\qangle{P, D', PT, BP, LSQ', Cache', TLB', E'}
}
\end{mathpar}
\begin{mathpar}
\small
\inferrule[None]{
\progreq(E)\!=\!\noneop()\\
E' \!=\! \update(E)
}{
\qangle{P, D, PT, BP, LSQ, Cache, TLB, E} \xrightarrow{\noneop()}\\
\qangle{P, D, PT, BP, LSQ, Cache, TLB, E'}
}
\end{mathpar}

In the baseline, the core components we wrapped as $E$ implicitly contain a security check on memory accesses. 
If $E$ tries to commit an instruction and encounters an exception, $E$ will stop sending new requests from the next step to resemble a crash.
As a result, both the request trace and the microarchitectural observation trace terminate.
In other words, a request trace with length $t$ implies that the machine does not crash in the first $t$ steps.

\pgheading{\sys Execution Semantics}
In \sys, virtual addresses are first converted to masked addresses and then translated to physical addresses.
The address translation procedure and address-indexed microarchitecture structures all use masked addresses as input.
We write down the operational semantics on \sys and highlight the parts that differ from the baseline.
\begin{mathpar}
\small
\inferrule[Fetch]{
\progreq(E)\!=\!\fetchop(v, v_\text{src}) \and
w\!=\!\addrmask(v)\\
w_\text{src}\!=\!\addrmask(v_\text{src})\\
BP' \!=\! \update(BP, \highlightmath{w}, \highlightmath{w_\text{src}}) \\
Cache' \!=\! \update(Cache, \highlightmath{\ptw(w, PT)}, \highlightmath{\trans(w, PT)}, TLB)\\
TLB' \!=\! \update(TLB, \highlightmath{w}, \highlightmath{\trans(w, PT)})\\
E' \!=\! \update(E, BP', LSQ, Cache', TLB', P[\highlightmath{\trans(w, PT)}])
}{
\qangle{P, D, PT, BP, LSQ, Cache, TLB, E} \xrightarrow{\fetchop(v, v_\text{src})}\\
\qangle{P, D, PT, BP', LSQ, Cache', TLB', E'}
}
\end{mathpar}
\begin{mathpar}
\small
\inferrule[Load]{
\progreq(E)\!=\!\loadop(v)\\
w\!=\!\addrmask(v)\\
LSQ' \!=\! \update(LSQ, \highlightmath{w}, \highlightmath{\trans(w, PT)}) \\
Cache' \!=\! \update(Cache, \highlightmath{\ptw(w, PT)}, \highlightmath{\trans(w, PT)}, TLB)\\
TLB' \!=\! \update(TLB, \highlightmath{w}, \highlightmath{\trans(w, PT)})\\
E' \!=\! \update(E, BP, LSQ', Cache', TLB', D[\highlightmath{\trans(w, PT)}])
}{
\qangle{P, D, PT, BP, LSQ, Cache, TLB, E} \xrightarrow{\loadop(v)}\\
\qangle{P, D, PT, BP, LSQ', Cache', TLB', E'}
}
\end{mathpar}
\begin{mathpar}
\small
\inferrule[Store]{
\progreq(E)\!=\!\storeop(v,d)\\
w\!=\!\addrmask(v)\\
D'\!=\!\update(D, \highlightmath{\trans(w, PT)}, d)\\
LSQ' \!=\! \update(LSQ, \highlightmath{w}, \highlightmath{\trans(w, PT)}) \\
Cache' \!=\! \update(Cache, \highlightmath{\ptw(w, PT)}, \highlightmath{\trans(w, PT)}, TLB)\\
TLB' \!=\! \update(TLB, \highlightmath{w}, \highlightmath{\trans(w, PT)})\\
E' \!=\! \update(E, BP, LSQ', Cache', TLB')
}{
\qangle{P, D, PT, BP, LSQ, Cache, TLB, E} \xrightarrow{\storeop(v,d)}\\
\qangle{P, D', PT, BP, LSQ', Cache', TLB', E'}
}
\end{mathpar}

We omit the semantics for the $\noneop$ request, since it is the same as in the baseline.
Furthermore, \sys adds the virtual address check request, denoted as $\checkop(v)$.
\sys stores the ASLR offset for $w\!=\!\addrmask(v)$ in the page table, so we define the operation to get the correct offset as $\getoffset(w, PT)$.
We then define the semantic for virtual address check as:
\begin{mathpar}
\small
\inferrule[Virtual-Address-Check]{
\progreq(E)\!=\!\checkop(v) \\
w\!=\!\addrmask(v)\\
v\!=\!\addrvirt(w, \getoffset(w, PT))\\
E' \!=\! \update(E, v)
}{
\qangle{P, D, PT, BP, LSQ, Cache, TLB, E} \xrightarrow{\checkop(v)}\\
\qangle{P, D, PT, BP, LSQ, Cache, TLB, E'}
}
\end{mathpar}

We do not need to have semantics for the cases where the above check fails with $v\!\neq\!\addrvirt(w, \getoffset(w, PT))$.
In our semantics, a crash is modeled as the termination of the execution trace. 
If the check fails, the machine will crash and cannot proceed to a new step of execution.

\subsection{Proof}
We begin by defining a notion of \textit{functional equivalence}.
Given two starting states initialized with the same program but different layouts, $S\!=\!\init(P, L)$ and $S'\!=\!\init(P, L')$, 
we say the execution of $S$ and $S'$ are functional equivalent for $t$ cycles,
if the two machines generate the same type of requests at each cycle and the virtual address in each request satisfies one of the following requirements: (1) the virtual addresses are the same, or (2) the virtual addresses are mapped to the same physical address.
\begin{definition}[$\progreq (S, t)\!\progeq\!\progreq(S',t)$]
\normalfont For $L$, $L'$, $P$, $S\!\coloneqq\! \init (P, L)$ and $S'\!\coloneqq\!  \init (P, L')$,
$\progreq(S, t)$ and $\progreq(S', t)$ are \emph{functional equivalent} if for all $k\!\in\! [0, t)$, $\progreq(S, t)[k]\!\coloneqq\! op(v, [v_\text{src}])$ and $\progreq(S', t)[k]\!\coloneqq\! op'(v', [v_\text{src}'])$, there is $op \!=\! op' \land  (v\!=\!v'  \lor L(v)\!=\!L'(v')\!\neq\! \bot)\land  (v_\text{src}\!=\!v_\text{src}' \lor L(v_\text{src})\!=\!L'(v_\text{src}')\!\neq\! \bot)$.
\end{definition}

We use functional equivalence as the assumption to prove our security property.
This assumption states that the ASLR secret is not used to calculate transmitter operands such as load/store addresses or branch targets, thereby not leaked via the third path in Section~\ref{sec:attack-analysis}.
In other words, regardless of the layouts chosen by ASLR, the machine speculatively or non-speculatively accesses the same virtual address or the same instruction/data (i.e., $v\!=\!v' \lor L(v)\!=\!L'(v')\!\neq\! \bot$).
Again, we assume the execution of the two states will not result in a crash.
The functional equivalence assumption also rules out the case when ASLR secrets are leaked via software attacks.

\pgheading{Reason About Attacks on the Baseline}
% \todo{Please check this part.}
When a program executes using the baseline execution semantics, even if it satisfies the functional equivalence property, its observable microarchitecture states leak the layout.
Consider resolving a branch under two layouts $L \!\neq\! L'$ with functional equivalent requests, i.e. $S_i\xrightarrow{\fetchop(v, v_\text{src})}S_{i\!+\!1}$, $S_i'\xrightarrow{\fetchop(v', v_\text{src}')}S_{i\!+\!1}'$, and $(v\!=\!v'\lor L(v)\!=\!L'(v')\!\neq\! \bot)\land(v_\text{src}\!=\!v_\text{src}'\lor L(v_\text{src})\!=\!L'(v_\text{src}')\!\neq\! \bot)$.
Here $v_\text{src},v_\text{src}'$ are the source branch addresses and $v, v'$ are the branch target addresses.

When $v\!=\!v'$, given two different layouts, we consider the case where $v$ is valid in the layout $L$ while invalid in the layout $L'$. Then, their translation results are different ($\trans(v, PT)\!\neq\!\bot\!= \trans(v', PT')$).
According to the execution semantics, using different translation results to update $Cache$ and $TLB$ will result in distinguishable microarchitecture observations.
The above description reassembles the first leakage path in Section~\ref{sec:attack-analysis},
where probing the same attacker-controlled address leads to different side effects due to the secret-dependent layout.

Consider the other case where $v\!\neq\! v'$ but $L(v)\!=\!L'(v')\!\neq\! \bot$.
This reassembles the second leakage path in Section~\ref{sec:attack-analysis} where $v$, $v'$ are code pointers for the same victim function under different layouts $L$, $L'$.
According to the baseline execution semantics, using different virtual addresses to update a bunch of microarchitecture structures, including $BP$ and $TLB$, results in distinguishable microarchitecture observations.

\pgheading{\sys Memory Layouts}
We formalize the valid layouts considered by \sys, denoted as $\syslayout$.
In this proof, we assume the code length is equal to the subregion size.
Given a physical program memory $P: [\pistart, \piend) \rightarrow Inst$ and an ASLR randomization region $[\vistart, \viend)$ in virtual memory, \sys divides the virtual memory region into $N$ sub-regions with equal size $\codelen$. 
\sys's layout set is defined as $\syslayout\!=\!\qty{L_i: i\!\in\! [0, N)]}$ such that
each $L_i$ maps the $i$-th subregion in the virtual memory $[\vistart\!+\!\loadoffset_i, \vistart\!+\!\loadoffset_i\!+\!\codelen)$ to the physical memory $[\pistart, \pistart\!+\!\codelen)$
where $\loadoffset_i\!=\! i\!\times\! \codelen$.

We define \textit{mask equivalence} to describe that two machines generate requests with the same masked addresses.
\begin{definition}[$\progreq(S, t)\!\maskeq\! \progreq(S', t)$]
\normalfont
Two traces $\progreq(S, t)$ and $\progreq(S', t)$ are \emph{mask equivalent}
if for all $k\!\in\! [0, t)$, $\progreq(S, t)[k]\!\coloneqq\! op(v, [v_\text{src}])$ and $\progreq(S', t)[k]\!\coloneqq\! op'(v', [v_\text{src}'])$, there is $op \!=\! op' \land \addrmask(v)\!=\!\addrmask(v') \land \addrmask(v_\text{src})\!=\!\addrmask(v_\text{src}')$.
\end{definition}
We prove the following lemma, which states that if two \sys state machines are initialized using different layouts from $\syslayout$ and they are functional equivalent, then they are also mask equivalent.

\begin{lemma}\label{lemma:prog-mask-eq}
\normalfont
    For all $L$, $L'\!\in\! \syslayout$, $P$, $S\!=\!\init_{\sys}(P, L)$, and $S'\!=\!\init_{\sys}(P,L')$, 
    \begin{align*}
    \progreq(S, t) \!\progeq\! \progreq(S', t) \Rightarrow \progreq(S, t) \!\maskeq\! \progreq(S', t)
    \end{align*}
\end{lemma}
\begin{proof}
    For all $k\!\in\! [0, t)$, we denote $\progreq(S, t)[k]\!=\!op(v,[v_\text{src}])$ and $\progreq(S', t)[k]\!=\!op'(v', [v_\text{src}'])$.
    Following the functional equivalence definition, if $\progreq(S, t)\!\progeq\! \progreq(S',t)$, there must be 
    \begin{equation}
    \begin{aligned}
    op\!=\!op' \, \land & (v \!=\! v'  \, \lor \, L(v)\!=\!L'(v')\!\neq\! \bot) \\
    \land & (v_\text{src} \!=\! v_\text{src}'  \, \lor \, L(v_\text{src})\!=\!L'(v_\text{src}')\!\neq\! \bot).
    \end{aligned}
    \label{eq:func-eq}
    \end{equation}
    Suppose the offsets used to define $L$ and $L'$ are $\loadoffset$ and $\loadoffset'$ respectively. 
    Consider the following two cases:
    \begin{itemize}[leftmargin=*]
        \item If $v\!=\!v'$, then $\addrmask(v)\!=\!\addrmask(v')$. 
        \item If $L(v)\!=\!L'(v')\!\neq\! \bot$, then there must exists $r\!\in\! [0, \codelen)$, which denotes the relative distance of the instruction in the program, such that 
        \begin{align*}
            L(v) \!=\! L(v') & \!=\! \pistart \!+\! r\\
            v & \!=\! \vistart\!+\!\loadoffset \!+\! r \\
            v' & \!=\! \vistart \!+\! \loadoffset' \!+\! r
        \end{align*}
        According to the definition of $\syslayout$, there must be 
        \begin{equation*}
            \loadoffset \bmod \codelen \!=\! \loadoffset'  \bmod \codelen \!=\! 0.
        \end{equation*}
        Then, using the definition of $\addrmask$, we get
        \begin{align*}
            \addrmask(v)\!=\! \vistart \!+\! r \!=\! \addrmask(v').
        \end{align*}
    \end{itemize}
    Similarly, condition (\ref{eq:func-eq}) also implies that $\addrmask(v_\text{src})=\addrmask(v_\text{src}')$.
    Therefore, we prove that if $\progreq(S, t) \!\progeq\! \progreq(S', t)$, we have $\progreq(S, t) \!\maskeq\! \progreq(S', t)$.
\end{proof}

Next, we prove the following lemma, 
which states that for all the possible \sys layouts, \sys generates the same address translation requests and output for the same masked address.
\begin{lemma}
\normalfont
    For all $L$, $L'\!\in\! \syslayout$, $P$, $S\!=\!\init_{\sys}(P, L)$, $S'\!=\!\init_{\sys}(P,L')$, 
    $PT$ and $PT'$ are page tables determined by $L$ and $L'$,
    there must be $\trans(w, PT)\!=\!\trans(w, PT')$ and $\ptw (w, PT)\!=\!\ptw(w, PT')$.
    \label{lemma:init-pt}
\end{lemma}
\begin{proof}
    We define a function to map masked addresses to physical addresses as $F(w)\!=\!L(\addrvirt(w))$.
    For any layout $L\!\in\! \syslayout$ with its corresponding $\addrvirt$ and $\loadoffset$, given a masked address $w$ pointing to an instruction and the relative distance of the instruction in the program as $r$, we have
    \begin{align*}
        F(w) & \!=\! L(\addrvirt(w)) \!=\! L(w \!+\! \loadoffset)\\
        & \!=\! L(\vistart \!+\! \loadoffset \!+\! r) \!=\! \pistart \!+\! r.
    \end{align*}
    The equation above shows that $F$ is independent of $L$ and $\loadoffset$, meaning that on \sys, for different layouts, the mapping between masked addresses and physical instruction addresses is always identical.
    Since \sys's page table translation and page table walk are both determined by this map, for all $L$, $L'\!\in\! \syslayout$, we have $\trans(w, PT)\!=\!\trans(w, PT')$ and $\ptw (w, PT)\!=\!\ptw(w, PT')$.
\end{proof}

For the convenience of describing induction assumption in our proof, we define public equivalence, a property stronger than microarchitectural indistinguishability.
\begin{definition}[$S\!\pubeq\! S'$]
\normalfont
    Given two machine states
    \begin{align*}
        S & \!=\!\qangle{P, D, PT, BP, LSQ, Cache, TLB, E}\\
        S' & \!=\!\qangle{P', D', PT', BP', LSQ', Cache', TLB', E'},
    \end{align*}
    they are public equivalent ($S\!\pubeq\! S'$), if $P\!=\!P'$, $\uarchobs(S)\!=\!\uarchobs(S')$, and for all $w\!\in\! [\vistart, \vistart\!+\!\codelen)$, $\trans(w, PT)\!=\!\trans(w, PT')$ and $\ptw (w, PT)\!=\!\ptw(w, PT')$.
\end{definition}

We continue to prove the lemma below, which states that if two executions are mask equivalent, then they are not micro-architecturally distinguishable.

\begin{lemma}\label{lemma:mask-non-interference}
\normalfont
For all layouts $L$, $L'\!\in\! \syslayout$, $P$, and two initial states $S\!=\!\init_{\sys}(P,L)$, $S'\!=\!\init_{\sys}(P,L')$, 
\begin{align*}
\progreq(S, t) \!\maskeq\! \progreq(S',t)  \Rightarrow \uarchobs(S, t) \!=\! \uarchobs(S', t)
    \end{align*}
\end{lemma}
\begin{proof}
    We use induction to prove that for all $k\!\in\! [0, t]$, $S_k\!\pubeq\! S_k'$, which is stronger than $\uarchobs(S, t) \!=\! \uarchobs(S', t)$. 
    We denote the state at cycle $k$ as $S_k \!=\!\qangle{P, D_k, PT_k, BP_k, LSQ_k, Cache_k, TLB_k, E_k}$.
    \begin{itemize}[leftmargin=*]
        \item \textbf{The base step ($S_0\!\pubeq\! S_0'$):} According to Lemma~\ref{lemma:init-pt}, we have $\trans(w, PT_0)=\trans(w, PT_0')$ and $\ptw(w, PT_0)=\ptw(w, PT_0')$. 
        Based on the initialization operation, we also have $\uarchobs(S_0)\!=\!\uarchobs(S_0')$. 
        Hence, $S_0\!\pubeq\! S_0'$.\\
        \item \textbf{The induction step ($S_k\!\pubeq\! S_k'\Rightarrow S_{k+1}\!\pubeq\! S_{k+1}'$):}
        By the induction assumption, we know for all $w$, $\trans(w, PT_{k})\!=\!\trans(w, PT_{k}')$ and $\ptw(w, PT_{k})\!=\!\ptw(w, PT_{k}')$.
        Since none of the operational semantics in \sys modifies $PT$, we have $PT_{k+1}\!=\!PT_k$ and $PT_{k+1}'\!=\!PT_k'$.
        As such, the address translation result ($\trans$) and requests generated by page table walks ($\ptw$) from the two machines are still the same at step $k\!+\!1$.\\
        
        Furthermore, by the mask equivalent assumption, we induce $\progreq(E_k)\!\maskeq\!\progreq(E_k')$.
        Following \sys's operational semantics, we derive that 
        for all requests $\fetchop, \loadop, \storeop, \noneop$, there must be $\uarchobs(S_{k\!+\!1})\!=\!\uarchobs(S_{k\!+\!1}')$.\\
        
        For the $\checkop$ request, 
        since both machines proceed one step in the cycle, the virtual-address-check must be successful on both machines, which implies that $\uarchobs(S_{k\!+\!1})\!=\!\uarchobs(S_{k\!+\!1}')$.
        Hence, $S_{k+1}\!\pubeq\! S_{k+1}'$.
    \end{itemize}
    Therefore, as the induction shows for all $k\!\in\! [0, t]$, $S_k\!\pubeq\! S_k'$, we conclude that $\uarchobs(S, t) \!=\! \uarchobs(S', t)$.
\end{proof}

Finally, we prove \sys achieves the microarchitectural indistinguishability property if the program and system satisfy the functionality equivalence assumption.
The theorem below can be directly derived using Lemma~\ref{lemma:prog-mask-eq} and Lemma~\ref{lemma:mask-non-interference}.

\begin{theorem}% [\textbf{\sys has an entropy of $\abs{\syslayout}$}]
\label{theorem:safe}
\normalfont
For all layouts $L$, $L'\!\in\! \syslayout$, $P$, $S\!=\!\init_{\sys}(P,L)$, and $S'\!=\!\init_{\sys}(P,L')$, 
\begin{align*}
\progreq(S, t) \!\progeq\! \progreq(S', t) \Rightarrow \uarchobs(S, t)  \!=\! \uarchobs(S', t).
    \end{align*}
\end{theorem}